%% file: paper.tex
\documentclass[conference]{IEEEtran}
\pdfoutput=1
\usepackage[numbers]{natbib}
\usepackage{amsmath,amsfonts,amssymb,graphicx,mathtools,flexisym, amsthm}
\usepackage{bbm}
\usepackage{algorithm,algorithmic}
\usepackage{booktabs}

\renewcommand*{\bibfont}{\small}

\usepackage{enumitem}
\newlist{myitemize}{itemize}{3}
\setlist[myitemize,1]{label=\textbullet,leftmargin=0.2in}
\setlist[myitemize,2]{label=$\rightarrow$,leftmargin=1em}
\setlist[myitemize,3]{label=$\diamond$}

\usepackage{multirow}
\usepackage{multicol}
\usepackage{booktabs}
\usepackage{makecell}
\usepackage{subfigure}

\usepackage{colortbl}
\usepackage{textpos}
\usepackage[dvipsnames]{xcolor}

\newtheorem{theorem}{Theorem}[section]
\newtheorem{defn}{Definition}[section]
\usepackage{url}
\newtheorem{definition}{Definition}

\makeatletter 
\DeclareRobustCommand*\cal{\@fontswitch\relax\mathcal}
\makeatother

\usepackage{etoolbox}
\makeatletter
\patchcmd{\@makecaption}
  {\scshape}
  {}
  {}
  {}
\makeatother

\begin{document}
%

\title{Scalable Explanation of Inferences on Large Graphs}


\author{
  \IEEEauthorblockN{
  Chao Chen\IEEEauthorrefmark{1}, 
  Yifei Liu\IEEEauthorrefmark{2}, 
  Xi Zhang\IEEEauthorrefmark{2}, 
  Sihong Xie\IEEEauthorrefmark{1}}
  \IEEEauthorblockA{\IEEEauthorrefmark{1}Department of Computer Science and Engineering, Lehigh University, Bethlehem, PA
    \\\{chc517, six316\}@lehigh.edu}
    \IEEEauthorblockA{\IEEEauthorrefmark{2}Beijing University of Posts and Telecommunications, Beijing, China
    \\\{liuyifei, zhangx\}@bupt.edu.cn }
    }
    
\maketitle

\begin{abstract}
Probabilistic inferences distill knowledge from graphs to aid human make important decisions.
Due to the inherent uncertainty in the model and the complexity of the knowledge,
it is desirable to help the end-users understand the inference outcomes.
Different from deep or high-dimensional parametric models,
the lack of interpretability in graphical models
is due to the cyclic and long-range dependencies
and the byzantine inference procedures.
Prior works did not tackle cycles and make \textit{the} inferences interpretable.
To close the gap,
we formulate the problem of explaining probabilistic inferences
as a constrained cross-entropy minimization problem
to find simple subgraphs that faithfully approximate the inferences to be explained.
We prove that the optimization is NP-hard, while the objective
is not monotonic and submodular to guarantee efficient greedy approximation.
We propose a general beam search algorithm
to find simple trees to enhance
the interpretability and diversity in the explanations,
with parallelization and a pruning strategy
to allow efficient search on large and dense graphs without hurting faithfulness.
We demonstrate superior performance on 10 networks from 4 distinct applications,
comparing favorably to other explanation methods.
Regarding the usability of the explanation,
we visualize the explanation in an interface that
allows the end-users to explore the diverse search results
and find more personalized and sensible explanations.
\end{abstract}


\input{introduction}
\input{prelim.tex}
\input{method.tex}
\input{experiment.tex}
\input{related_work.tex}

%
\bibliographystyle{plain}
\bibliography{paper.bbl}

\end{document}

%% file: introduction.tex
\section{Introduction}
Distilling knowledge from graphs
is an important task found ubiquitously in applications,
such as fraud detection~\cite{rayana2015collective},
user interest modeling in social networks~\cite{tang2009relational, tang2011learning, namata2011collective},
and bioinformatics~\cite{li2007protein,huang2011efficient,zitnik2017predicting}.
The knowledge helps humans make high-stake decisions, such as whether to investigate a business or account for fraud detection on Yelp or to conduct expensive experiments on a promising protein in drug discovery.
The state-of-the-art approaches model the graphs as directed or undirected graphical 
models, such as Bayesian networks, Markov Random Fields (MRF)~\cite{kindermann1980markov}, and Conditional Random Fields (CRF)~\cite{sutton2012introduction},
and distill knowledge from the graphs using inferences based on
optimization~\cite{yanover2006linear,hazan2010norm} and message passing~\cite{jo2018fast}.
Different from predictive models on i.i.d. vectors,
graphical models capture the dependencies among random variables and carry more insights for decision making.
However, the inferences involve iterative and recursive computations,
making the inference outcomes cognitively difficult to understand, verify, and ratify,
locking away more applications of graphical models
(EU law requires algorithmic transparency~\cite{goodman2017european})
and more accurate models through debugging~\cite{wachter2017counterfactual}.
We focus on MRFs and belief propagation (BP) inferences~\cite{koller2009probabilistic} that compute marginal distributions,
aiming to make the inference outcomes more interpretable and thus cognitively easier for humans.
Fig.~\ref{fig:bp_example} depicts the problem definition and the proposed solution. 



Several challenges are due.
First, simple but faithful explanations are desired~\cite{lombrozo2006structure}
but
have not been defined for inferences on MRFs.
Prior work~\cite{honorio2012variable,friedman2008sparse}
approximates a high-dimensional Gaussian through a sparse covariance matrix, which does not explain belief propagation.
To explain the marginal distribution on MRFs, using sensitivity analysis, the authors in~\cite{chan2005sensitivity,darwiche2003differential}
proposed to find influential parameters inherent in the model
but not in any particular inference algorithms and computations.
Explanation of parametric linear models and deep networks using surrogate models, differentiation, and feature selection~\cite{ribeiro2016should,koh2017understanding,guyon2003introduction,lou2012intelligible,caruana1999case} cannot be applied to graphical model inferences,
although our proposal can explain inferences on deep graphical models~\cite{johnson2016composing}.
Explainable RNNs~\cite{lei2016rationalizing} handles linear chains but not general MRFs.
In terms of usability,
previous works have studied how to visualize explanations of other models and utilize the explanations in the end tasks such as model debugging~\cite{stumpf2008integrating}.
It is less know what's the best \textit{graph} explanation complexity and faithfulness trade-off for the end-users, how to effectively communicate the probabilistic explanations, and how to utilize the explanations.

\begin{figure}[t]
    \centering
    \includegraphics[width=0.5\textwidth]{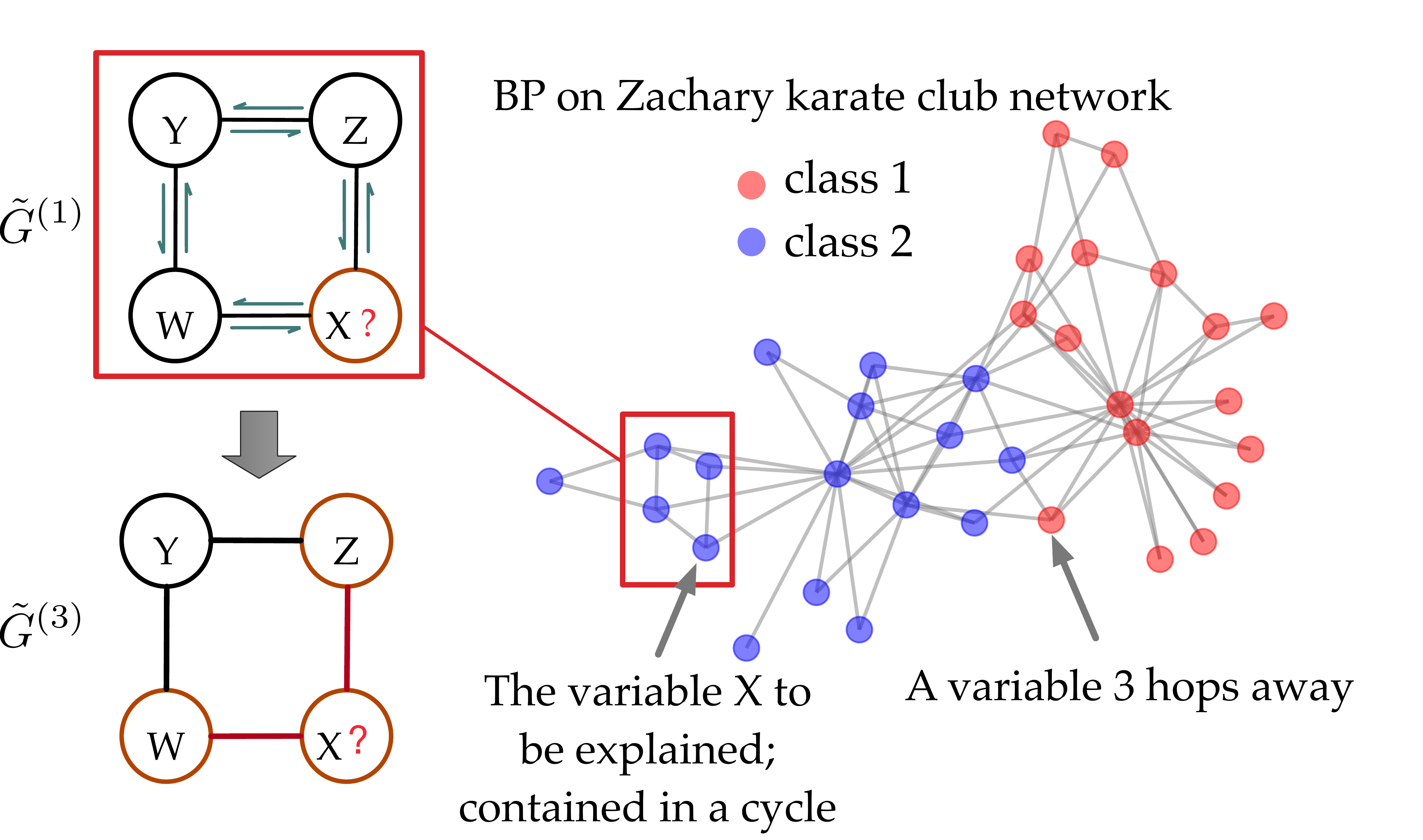}
    \caption{A cyclic graphical model $G$ for the Zachary karate club network,
    with BP inference outcomes shown in two colors.
    We focus on explaining how BP calculates the belief on $X$, highlighted in the subgraph $\tilde{G}^{(1)}$.
    Due to the cycles and long-range dependencies on $G$,
    a complete explanation is recursive and long.
    With messages, beliefs, and priors, GraphExp extracts a limited-size \textit{tree} $\tilde{G}^{(3)}$, on which $X$ has a marginal similar to that on $G$.
    }
    \label{fig:bp_example}
    \vspace{-.2in}
\end{figure}
Second, algorithmically,
an MRF can be large, densely connected, and cyclic, while
simple and faithful explanations need to be found efficiently.
BP computes each message using other messages
iteratively and recursively until convergence (Eq. (\ref{eq:sum_prod})).
As a result, a message is a function of other messages
and a complete explanation requires the entire history of message computations,
possibly from multiple iterations.
Furthermore, on cyclic graphs,
a converged message can be defined recursively by itself,
generating explanations such as ``a user is fraudulent because it is fraudulent''.
A desirable explanation should be without recursion,
but cutting out the cycles may result in a different model and affect faithfulness.



We propose a new approach called ``GraphExp'' to address the above challenges.
Given the full graphical model $G$ and any target variable $X$ (the ``explanandum'') on $G$, 
GraphExp finds an ``explanans''
(or the ``explaining'') graphical model $\tilde{G}$
consisting of a smaller number of random variables and dependencies.
The goal of GraphExp is to
minimize the loss in faithfulness measured by
the symmetric KL-divergence between the marginals of $X$ inferred on $\tilde{G}$ and $G$~\cite{suermondt1993explanation}.
Starting from the graph $\tilde{G}^{(1)}$ consisting of $X$ only,
GraphExp greedily includes the next best variable into the previous subgraph
so that the enlarged subgraph has the lowest loss.
Theoretically we prove that:
\textbf{(1)} an exhaustive search for the optimal
$\tilde{G}$ with highest faithfulness (lowest loss) is NP-hard,
and furthermore,
the objective function is neither monotonic nor submodular, leading to the lack of a performance guarantee of any greedy approximation (Theorem~\ref{thm:submodular});
\textbf{(2)} GraphExp only generates acyclic graphs that are more explainable (Theorem~\ref{thm:trees}).

There can exist multiple sensible explanations for the same inference outcome~\cite{russell2019efficient,ross2017right} and an end-user can find the one that best fits her mental model.
We equip GraphExp with beam search~\cite{batra2012diverse} to discover
a set of distinct, simple, and faithful explanations for a target variable.
Regarding scalability, when looking for $\tilde{G}$ on densely connected graphs,
the branching factor in the search tree can be too large for fast search.
While the search is trivially parallelizable,
we further propose a safe pruning strategy that retains the desirable candidates
while cutting the search space down significantly (Fig.~\ref{fig:speedup}).
Regarding usability, GraphExp does not commit to a single explanation
but allows the users to select one or multiple most sensible explanations for further investigation (Section~\ref{sec:user_study}).
We highlight the contributions as follows:
\begin{myitemize}
\item We define the problem of explaining belief propagation
to end-users and study the communication and utility of the explanations.
\item We propose an optimization problem for finding simple, faithful, and diverse explanations.
We prove the hardness of the problem
and then propose GraphExp as a greedy approximation algorithm to solve the problem.
We analyze the time complexity and the properties of the output graphs,
Both parallel search and the proposed pruning strategy deliver at least linear speed-up on large datasets.
\item Empirically,
on 10 networks with up to millions of nodes and edges from 4 domains,
GraphExp explains BP faithfully
and significantly outperforms variants of GraphExp and other explanation methods not designed for graphs.
We propose visualization that allows flexible, intuitive of the found explanations.
To demonstrate the utility of the found explanations, we identify a security issue in Yelp spam detection using the found subgraphs.
\end{myitemize}
\begin{table}[t]
    \centering
    \small
    \caption{Comparing GraphExp with prior explanation methods: FS (Feature Selection), LIME, GSparse (graph sparsification), and GDiff (graph differentiation).
    ($\ast$: ``surely yes''; $\circ$: ``partially''; emptiness: ``no'').
    }
    \begin{tabular}{c||cccc|c}
                    & \rotatebox[origin=c]{90}{FS~\cite{guyon2003introduction}}
                    & \rotatebox[origin=c]{90}{LIME~\cite{ribeiro2016should}}
                    & \rotatebox[origin=c]{90}{GSparse~\cite{friedman2008sparse}}
                    & \rotatebox[origin=c]{90}{GDiff~\cite{chan2005sensitivity}}
                    & \rotatebox[origin=c]{90}{GraphExp}\\\hline
        Cycles handling & & & $\ast$ & $\ast$ & $\ast$ \\
        Completeness & & & $\circ$ & $\ast$ & $\ast$ \\
        Interpretability & $\circ$ & $\ast$ & $\circ$ & $\circ$ & $\ast$ \\
        Diversity & $\circ$ & $\circ$ & & & $\ast$ \\
        Scalability & & $\circ$ & & $\circ$ & $\ast$ \\
        Flexibility & & $\circ$ & & $\circ$ & $\ast$ \\
    \end{tabular}
    \label{tab:differences}
    \vspace{-.2in}
\end{table}

%% file: prelim.tex
\section{Problem definition}
\label{sec:prelim}
Notation definitions are summarized in Table~\ref{tab:definitions}.
\begin{table}[t]
    \caption{ Notation Definitions }
    \label{tab:definitions}
    \centering
    \small
    \begin{tabular}{|c||c|}\hline
        Notation & Definition\\\hline
        $G=(V, E)$ & Undirected graphical model (MRF)\\
        $V, E$ & Random varaibles and their connections\\
        $X_i, X, Y$ ($x_i, x, y$) & Random variables (and their values)\\
        $\phi_X(x)$ (or $\phi_i(x_i)$) & Prior probability distribution of $X$ (or $X_i$)\\
        $\psi_{XY}(x,y)$ & Compatibility matrix between $X$ and $Y$\\
        $m_{X\to Y}(y)$ & Message passed from $X$ to $Y$\\
        $b_X(x)\triangleq P(X=x)$ & Marginal distribution (belief) of $X$\\
        KL$(p||q)$ & KL Divergence between $p$ and $q$\\
        $\partial G^\prime$ & Variables in $G\setminus G^\prime$ connected to subgraph $G^\prime$\\
        ${\cal N}(X_i)$ & Neighbors of $X_i$ on $G$\\
        \hline
    \end{tabular}
\end{table}
Given a set of $n$ random variables $V=\{X_1,\dots,X_n\}$,
each taking values in $\{1,\dots, c\}$ where $c$ is the number of classes,
an MRF $G$
factorizes the joint distribution $P(V)$ as
\begin{equation}
\label{eq:joint}
P(V)=\frac{1}{Z}\prod_{i\in [n]} \phi_i(X_i) \prod_{i,j\in [n]}\psi_{ij}(X_i, X_j),
\end{equation}
where $Z$ normalizes the product to a probability distribution, $\phi_i$ is the prior distribution of $X_i$ without considering other variables.
The compatibility $\psi_{ij}(x_i, x_j)$ encodes how likely the pair
$(X_i, X_j)$ will take the value $(x_i, x_j)$ jointly
and capture the dependencies between variables.
The factorization can be represented by a graph consisting of
$X_1,\dots, X_n$ as nodes and edges $(X_i, X_j)$, as shown in Fig.~\ref{fig:bp_example}.
BP inference computes the marginal distributions (beliefs) $b_X$, $X\in V$,
based on which human decisions can be made.
The inference computes messages $m_{i\to j}(x_j)$ from $X_i$ to $X_j$:
\begin{equation}
\label{eq:sum_prod}
\frac{1}{Z_j}\sum_{x_i}
\left[\psi_{ij}(x_i, x_j)\phi_i(x_i) \prod_{k\in {\cal N}(X_i)\setminus \{j\}} m_{k\to i}(x_i)\right],
\end{equation}
where $Z_j$ is a normalization factor so that $m_{i\to j}$ is a probability distribution of $x_j$.
The messages in both directions on all edges $(i,j)\in E$ are updated  until convergence (guaranteed when $G$ is acyclic~\cite{pearl2014probabilistic}).
The marginal $b_X$ is the belief
\begin{equation}
\label{eq:belief}
    b(X) \propto \phi(X)\prod_{Y\in {\cal N}(X)} m_{Y\to X}(X),
\end{equation}
where ${\cal N}(X)$ is the neighbors of $X$ on $G$.
We aim to explain how the marginal is inferred by BP.
For any $X\in V$,
$b_X$ depends on messages over all edges reachable from $X$
and to completely explain how $b_X$ is computed,
one has to trace down each message in Eq. (\ref{eq:belief}) and Eq. (\ref{eq:sum_prod}).
Such a complete explanation is hardly interpretable due to two factors:
1) on large graphs with long-range dependencies, messages and variables far away from
$X$ will contribute to $b_X$ indirectly through many steps;
2) when there is a cycle,
BP needs multiple iterations to converge and a message can be recursively defined by itself~\cite{kang2010inference,jo2018fast}.
A complete explanation of BP can keep track of all these computations on a call graph~\cite{ryder1979constructing}.
However, the graph easily becomes too large to be intuitive for humans to interpret or analyze.
Rather, $b_X$ should be approximated using short-range dependencies
without iterative and recursive computations.
The question is, without completely follow the original BP computations,
how will the approximation affected?
To answer this question,
we formulate the following optimization problem:
\begin{definition}
Given an MRF $G$, and a target node $X\in V$,
extract another MRF $\tilde{G}\subset G$ with 
$X\in \tilde{G}$ and containing no more than $C$ variables and no cycle,
so that BP computes similar marginals $b_X$ and $\tilde{b}_X$
on $G$ and $\tilde{G}$, respectively.
Formally, solve the following
\begin{equation}
\label{eq:opt}
\arraycolsep=1.4pt\def\arraystretch{1.5}
    \begin{array}{cc}
         \min\limits_{\tilde{G}} & d(b_X, \tilde{b}_X) 
        =\textnormal{KL}(b_X||\tilde{b}_X) +
        \textnormal{KL}(\tilde{b}_X||b_X) \\
         \textnormal{s.t.} & \tilde{G}\subset G, 
         \hspace{.1in}
         |\tilde{G}|\leq C,
         \hspace{.1in}
         X \in \tilde{G},
         \hspace{.1in}
         \tilde{G} \textnormal{ acyclic}.
    \end{array}
\end{equation}
\end{definition}
The objective captures the faithfulness of $\tilde{G}$, measured by the symmetric KL-divergence $d$ between marginal distributions of $X$ on $G$ and $\tilde{G}$, where
\begin{equation}
\label{eq:kl-div}
\textnormal{KL}(P||Q)= \sum_{x=1}^{c} P(x)\log[P(x) / Q(x)].
\end{equation}
The choice of $d$ as a faithfulness measured can be justified:
KL$(b_X||\tilde{b}_X)$ measures the loss when the ``true'' distribution is $b_X$ while the explaining distribution is $\tilde{b}_X$~\cite{suermondt1993explanation}.
Symmetrically, a user can regard $\tilde{b}_X$ as the ``true'' marginal, which can be explained by the $b_X$.
The simplicity of $\tilde{G}$ can be measured
by the number of variables on $\tilde{G}$,
and
for $\tilde{G}$ to be interpretable,
we control the size of $\tilde{G}$ to be less than $C$.
Since a graphical model encodes a joint distribution of a set of variables,
the above problem is equivalent to searching a joint distribution of a smaller number of variables with fewer variable dependencies,
so that the two joint distributions lead to similar marginals of the variable $X$ to be explained.

If the negation of the above objective function is submodular and monotonically increasing,
then a greedy algorithm that iteratively builds $\tilde{G}$ by adding one variable at a time to increase the negation of objective function (namely, to decrease $d(b_X, \tilde{b}_X)$)
can generate a solution whose value is within $(1-1/e)$ of the optimum~\cite{nemhauser1978analysis}.
\begin{defn}[Submodularity]
Let ${\cal V}$ be a set and $2^{\cal V}$ be the power set of ${\cal V}$.
A set function $f:2^{\cal V}\to \mathbb{R}$ is submodular if
$\forall {\cal A}\subset {\cal A}^\prime\subset {\cal V}$ and any $Y\not\in {\cal A}^\prime$, $f({\cal A}\cup \{Y\})-f({\cal A})\geq f({\cal A}^\prime \cup \{Y\}) - f({\cal A}^\prime)$.
\end{defn}
\begin{defn}
A set function $f:2^{\cal V}\to \mathbb{R}$
is monotonically increasing if
$\forall {\cal A}\subset {\cal A}^\prime\subset {\cal V}$, implies $f({\cal A})\leq f({\cal A}^\prime)$.
\end{defn}
\noindent
However, we prove that both properties do not hold on $d(b_X, \tilde{b}_X)$, which cannot be efficiently approximated well.

\begin{theorem}
\label{thm:submodular}
The objective function in Eq. (\ref{eq:opt}) is not submodular nor monotonically increasing.
\end{theorem}
\begin{proof}
We find a counterexample that violates submodularity and monotonicity.
As shown in Fig.~\ref{fig:submod},
the full graph $G$ has 3 variables $X$, $Y$ and $Z$, with $X$ connected to $Y$ and $Z$, respectively.
The priors are $\phi(X)=[0.5, 0.5]$, $\phi(Y)=[0.8, 0.2]$,
and $\phi(Z)=[0.1, 0.9]$, and both edges have 
the same homophily-encouraging potential $\psi_{ij}(a, b)=0.99$ if $a=b$ and $0.01$ otherwise.

Let subgraphs $G_0=\{X\}$, $G_1=\{X, Z\}$ and $G_2=\{X, Y\}$
be as shown in the figure.
One can run BP on $G_i,i=0,1,2$ to find $\tilde{b}_X$ and that
$(-d_{G_2}(b_X, \tilde{b}_{X}))-(-d_{G_0}(b_X, \tilde{b}_{X}))<
(-d_{G}(b_X, \tilde{b}_{X}))-(-d_{G_1}(b_X, \tilde{b}_{X}))$,
with the subscriptions $G_i$ indicating on which subgraph is $\tilde{b}_{X}$ computed.
However, $G_0\subset G_1$ and the gain through adding $Y$ to $G_1$ is greater than that adding $Y$ to $G_0$.
On the same example, we can see that adding $Y$ to $G_0$ can increase the objective. 
\end{proof}

\begin{figure}[t]
    \centering
    \includegraphics[width=0.4\textwidth]{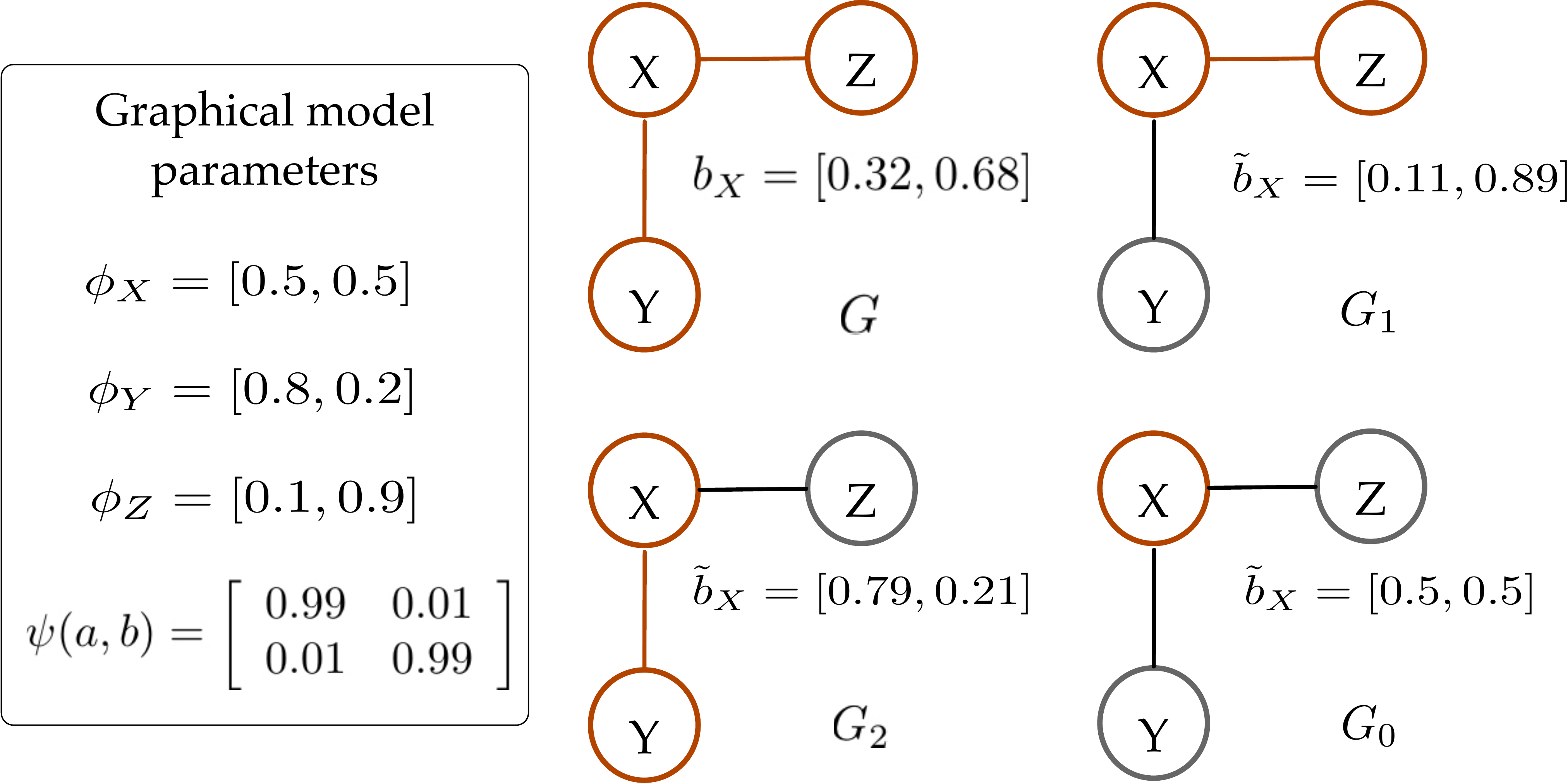}
    \caption{A graphical model on which submodularity and monotonicity of the objective function in Eq. (\ref{eq:opt}) do not hold.
    The goal is to find the best subgraph from $\{G_0, G_1, G_2\}$ that best
    approximates the belief of $X$ inferred on $G$.
    Edges and nodes in red are those included in the individual subgraphs.
    }
    \label{fig:submod}
    \vspace{-.2in}
\end{figure}

%% file: method.tex
\section{Methodologies}
\label{sec:method}
The optimization problem Eq. (\ref{eq:opt}) can be solved by exhaustive search in the space of all possible trees under the specified constraints and 
it is thus a combinatorial subset maximization problem and NP-hard~\cite{guyon2003introduction,amaldi1998approximability}, similar to exhaustive feature selection.
Greedy algorithms are a common solution to approximately solve NP-hard problems.
Since finding multiple alternative sensible explanations is one of our goals,
we adopt beam search in the greedy search~\cite{batra2012diverse}, maintaining in a beam several top candidates ranked by faithfulness and succinctness
throughout the search.

\begin{algorithm}
\caption{GraphExp (the general search framework)}
\label{alg:beam_search}
\begin{algorithmic}
\STATE \textbf{Input}:a graphical model $G = (V,E)$; a target random variable $X\in V$ to be explained;
\STATE prior $\phi_i$ and belief $b_i$, $\forall X_i \in V$; messages $m_{i\to j}$ for $\forall (X_i,X_j) \in E$;
\STATE maximum subgraph complexity $C$; beam size $k$.
\STATE \textbf{Output}:
Multiple explaining subgraphs $\tilde{G}$ for $X$, along with approximated computations of $b_X$.
\STATE \textbf{Init:} $\tilde{G}^{(1)}=(\tilde{V}, \tilde{E})$, $\tilde{V}=\{X\},\tilde{E}=\emptyset$. \texttt{Beam[1]}=$\{\tilde{G}^{(1)}\}$
\FOR{ $t=2\to C $ }
    \STATE \texttt{Beam[t]}$=\emptyset$.
    \FOR{each subgraph $\tilde{G}^{(t-1)}$ in \texttt{Beam[t-1]}}
    \STATE Find and add the top $k$ extensions of $\tilde{G}^{(t-1)}$ to \texttt{Beam[t]}.
    \ENDFOR
    \STATE Retain the top $k$ candidates in \texttt{Beam[t]}.
\ENDFOR
\STATE Run BP on each candidate graph in \texttt{Beam[C]} and obtain converged messages and beliefs as an approximation of computation of $b_X$ on $G$.
\end{algorithmic}
\end{algorithm}
A general greedy beam search framework is presented in Algorithm~\ref{alg:beam_search}.
The algorithm 
finds multiple explaining subgraphs $\tilde{G}^{(C)}$ of size $C$ in $C-1$ iterations, where $C$ is the maximum subgraph size
(which roughly equals to human working memory capacity~\cite{miller1956magical}, or the number of cognitive chunks~\cite{lage2019evaluation}).
Starting from the initial $\tilde{G}^{(1)}=\{X\}$,
at each step $t$ the graph $\tilde{G}^{(t-1)}$
is extended to $\tilde{G}^{(t)}$ by adding one more explaining node and edge to optimize a certain objective function without forming a loop.
After the desired subgraphs are found and right before the algorithm exits,
BP will be run again on $\tilde{G}$ to compute
$\tilde{b}_X$ so that we can use the converged messages on $\tilde{G}$
to explain to an end-user how $b_X$ is approximated on $\tilde{G}$.
Since $\tilde{G}$ is small and contains no cycle,
the explanation is significantly simpler than the original computations on $G$.
We substantiate the general framework in the following two sections, with two alternative ways to rank candidate extensions of the subgraphs during the beam search.
Before we discuss the two concrete algorithms,
a theoretical aspect related to interpretability is characterized below.
\begin{theorem}
\label{thm:trees}
The output $\tilde{G}$ from Algorithm~\ref{alg:beam_search} is a tree.
\end{theorem}
\begin{proof}
    We prove this by induction.
    Base case: $\tilde{G}^{(0)}$ is a single node so it is a tree.
    Inductive step: Assume that $\tilde{G}^{(t)}$ is a tree.
    By adding one more variable and edge ($X^\ast$ and $e^\ast$) to $\tilde{G}^{(t)}$, there is no cycle created, since we are only attaching $X^\ast$ to $\tilde{G}^{(t)}$ through a single edge $e^\ast$.
\end{proof}

\begin{figure}[t]
    \centering
    \includegraphics[width=0.47\textwidth]{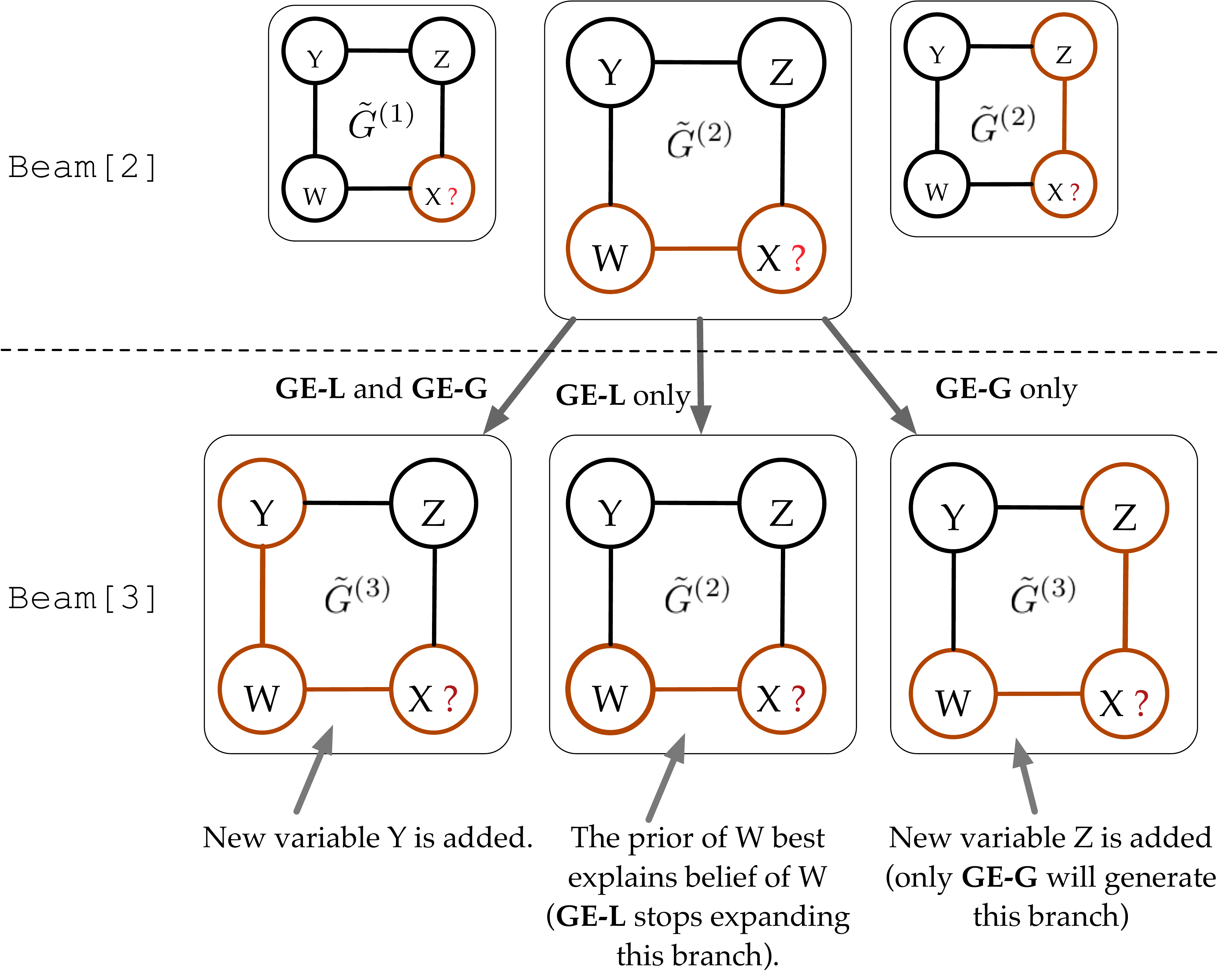}
    \caption{
    Explaining BP on a graphical model with four variables ($X$, $Y$, $Z$, and $W$),
    four edges, and a cycle.
    We show one step of beam search,
    extending a candidate $\tilde{G}^{(2)}$ in \texttt{Beam[2]} 
    to three candidates of $\tilde{G}^{(3)}$ in \texttt{Beam[3]} (with beam size 3).
    The rightmost candidate in \texttt{Beam[2]} is also extended but the extensions are not shown to avoid clutters.
    After all extensions are generated, only the top 3 are kept in \texttt{Beam[3]}.
    The edges and variables in red are included in the explaining subgraphs.
    Note that \textbf{GE-G} and \textbf{GE-L} have different behaviors when dealing with the extensions.
    }
    \label{fig:beam_search}
    \vspace{-.2in}
\end{figure}

\subsection{GraphExp-Global (\textbf{GE-G}): search explanations via evaluating entire subgraphs}
\label{sec:geg}
We propose \textbf{GE-G},
an instantiation of Algorithm~\ref{alg:beam_search}.
At iteration $t$ $(t=2,\dots, C)$, the algorithm evaluates
the candidate extensions of $\tilde{G}^{(t-1)}$ 
using the objective Eq. (\ref{eq:opt}).
Define $\partial \tilde{G}^{(t-1)}$
be the set of nodes in $G\setminus \tilde{G}^{(t-1)}$ that are connected to $\tilde{G}^{(t-1)}$.
A candidate $\tilde{G}^{(t)}$ is generated by adding a variable $Y\in \partial\tilde{G}^{(t-1)}$
through the edge $(Y,W)$ to $\tilde{G}^{(t-1)}$, where $W$ is a random variable in $\tilde{G}^{(t-1)}$.
A new BP procedure is run on $\tilde{G}^{(t)}$ to infer $\tilde{b}_X$, the marginal of $X$, and the distance $d(b_X,\tilde{b}_X)$ is calculated as the quality of the candidate $\tilde{G}^{(t)}$.
After exhausting all possible candidates,
\textbf{GE-G} adds the $k$ candidates with the least distances to \texttt{Beam[t]}.

One search step in Fig.~\ref{fig:beam_search} demonstrates the above ideas.
The search attempts to extend the middle $\tilde{G}^{(2)}$ in \texttt{Beam[2]}
by adding a new variable from $\partial{\tilde{G}^{(2)}} = \{Y, Z\}$ to the subgraph,
so the new belief $\tilde{b}_X$ on the larger subgraph is closer to $b_X$.
In this example,
\textbf{GE-G} keeps the top 3 extensions of $\tilde{G}^{(2)}$ (beam size is 3),
but only two options are legitimate and are the only two candidates included in \texttt{Beam[3]}. 
The middle subgraph consisting of $\{W, X\}$ is generated by the algorithm \textbf{GE-L} to be discussed later (see Section~\ref{sec:gel}).
When the search extends the bottom right subgraph $\tilde{G}^{(3)}$,
variable $Y\in \partial \tilde{G}^{(3)}$ can be connected to $\tilde{G}^{(3)}$ in two ways, through edges $Y\to Z$ and $Y\to W$, but both GraphExp variants include only one link to avoid cycles in the exlaining subgraphs.

On the high-level, \textbf{GE-G} is similar to forward wrapper-style feature selection algorithms,
where each feature is evaluated by including it to the set of selected features and running a target classification model on the new feature sets.
The key difference here is that {GE-G} can't select any variable on $G$,
but has to restrict itself to those that will result in an acyclic graph (which is guaranteed by Theorem~\ref{thm:trees}).

One of the advantages of \textbf{GE-G} is that the objective function in the optimization problem Eq.~(\ref{eq:opt}) is minimized directly at each greedy step.
However, as each candidate at each step requires a separate BP,
it can be time-consuming.
We analyze the time complexity below.
To generate a subgraph of maximum size $C$ for a variable,
$C-1$ iterations are needed.
At iteration $t=2,\dots, C$, one has to run BP for as many times as the number of neighboring nodes of the current explanation $\tilde{G}^{(t-1)}$.
The number of candidates that need to be evaluated for one of the $k$ candidates in \texttt{Beam[t-1]} in the iteration equals the size of the number of neighboring nodes in $\partial{\tilde{G}}^{(t-1)}$.
On graphs with a small diameter,
this size grows quickly to the number of nodes in $G$. 
On the other extreme, if $G$ is a linear chain,
this size is no more than 2.
For each BP run, it is known that BP will converge in the number of iterations same as the diameter of the graph that BP is operated on,
which is upper-bounded by the size of the candidate subgraph $\tilde{G}^{(t)}$.
During each BP iteration, $O(t)$ messages have to be computed.
The overall time complexity of \textbf{GE-G} is
$O(|V|k\sum_{t=2}^{C} |\partial \tilde{G}^{(t-1)}| t^2)$, where $k$ is the beam size.
Since the number of classes on the variables are fixed and usually small (relative to the graph size), here we ignore the factor $O(c^2)$, which is the time complexity to compute one message.

\subsection{Speeding up \textbf{GE-G} on large graphs}
Graphical models in real-world applications are usually gigantic
containing tens or hundreds of thousands of nodes.
GraphExp can take a long time to finish on such large graphs,
especially when the graph diameter is small.
Slowness can harm the usability of GraphExp in applications requiring interpretability,
for example,
when a user wants to inspect multiple explanations of BP for counterfactual analysis, or statistics of the errors needs to be computed over explanations of many nodes~\cite{wachter2017counterfactual}.
We propose parallelized search and a pruning strategy to speed up \textbf{GE-G}.

\noindent\textbf{Parallel search}
The general GraphExp algorithm
can be parallelized on two levels.
First, the generation of explanations over multiple target variables can be executed on multiple cores.
Second, in the evaluation of the next extensions of $\tilde{G}^{(t-1)}$ during beam search,
multiple candidates $\tilde{G}^{(t)}$ can be tried out at the same time on multiple cores.
Particularly for \textbf{GE-G},
during the BP inference over each candidate $\tilde{G}^{(t)}$,
there are existing parallel algorithms that compute the messages~\cite{gonzalez2009residual} asynchronously.
As the subgraphs $\tilde{G}^{(t)}$ are bounded by the human cognitive capacity and are small, a parallel inference can be an overkill.
We evaluate the reduction in search time using the first level of parallelism (Section~\ref{sec:scalability}).



\noindent\textbf{Pruning candidate variables}
In Algorithm~\ref{alg:beam_search}, all candidates $e\in C=(G\setminus\partial{\tilde{G}^{(t)}}, \partial{\tilde{G}^{(t)}})$ have to be evaluated to select $(X^\ast, e^\ast)$ and we have to run BP as many times as $|C|$.
When the cut $C$ is large, this can be costly. 
As we aim at explaining how BP infers the marginal of the target $Y$,
adding any variable that has a distribution that deviates much from the distribution of $Y$ is not helpful but confusing.
Considering that a subgraph ${\tilde{G}^{(t-1)}}$ has $|C|$ candidates at step-$t$, we run BP on these $|C|$ candidates and abandon the bottom $(100-p)$ percent of them based on Eq. \ref{eq:opt} in the following steps. 

\subsection{GraphExp-Local (\textbf{GE-L}): search explanations via local message back-tracing}
\label{sec:gel}
Sometimes one may want to trade explanation faithfulness for speed during subgraph search.
For example, in the exploratory phase,
a user wants to get a quick feeling of how the inferences are made or to identify mistakes caused by glitches on the graph before digging deeper into finer-grained explanations.
We propose \textbf{GE-L} for this purpose to complement \textbf{GE-G} that can generate more faithful and detailed explanations (with the expense of more searching time).
\textbf{GE-L} is based on message back-tracing that follows the general beam search but with more constraints on the search space.
At iteration $t$, the search adds an additional edge between $\tilde{G}^{(t-1)}$ and $\partial\tilde{G}^{(t-1)}$
that best explains a message or a belief in
$\tilde{G}^{(t-1)}$, which can be computed using 
information that is not currently in $\tilde{G}^{(t-1)}$.
There are two cases.
\begin{itemize}
    \item 
For a message $m_{Y\to W}$ on an edge $(Y, W)$ already included in $\tilde{G}^{(t-1)}$,
the search attempts to find a message $m_{Z\to Y}$,
where $Z\in \partial\tilde{G}^{(t-1)}$,
so that the message $m_{Z\to Y}$ contributes most to the message $m_{Y\to W}$.
We use the distance $d$ defined in Eq. (\ref{eq:opt}) to measure the contribution of $m_{Z\to Y}$ to $m_{Y\to W}$: the smaller the distance,
the more similar two messages are and thus more contribution from $m_{Z\to Y}$
to $m_{Y\to W}$.
\item
For the belief $b_X$ of the target node $X$,
the search attempts to find a message $m_{W\to X}$ that best explains $b_X$,
using the distance between $m_{W\to X}$ and $b_X$.
\end{itemize}
In both cases,
we define the \textbf{end points} of $\tilde{G}^{(t-1)}$ as those nodes that are in $\tilde{G}^{(t-1)}$ but can be connected to those in $\partial\tilde{G}^{(t-1)}$.
In the example in Fig.~\ref{fig:bp_example},
$Y$ and $X$ are end points of the subgraph $\tilde{G}^{(2)}$ in the middle. 
In \textbf{GE-L},
if the prior of an end-point
best explains the message emitting from the end-point (e.g., $m_{Y\to X}$ in Fig.~\ref{fig:bp_example}) or belief of the end-point ($b_X$ in the same figure),
the prior is added to the subgraph and no extension can be made to the end point:
the prior is a parameter of the graphical model and not computed by BP,
and no other BP distribution can further explain the prior.
The search of \textbf{GE-L} will stop at this branch, although the other branches in the beam can further be extended.
Using the same example in Fig.~\ref{fig:bp_example},
the prior $\phi_{W}$ best explains $m_{W\to X}$) and
this branch is considered finished requiring no more extension.


\textbf{Analysis of \textbf{GE-L}}
We analyze what the above search is optimizing.
We unify the two cases where the search
explains how a message and belief are computed.
Assuming that the potential function $\psi_{ij}$ for all edges $(X_i, X_j)$ are
identity matrices, which is the case when a graph exhibits strong homophily.
Then the message going from $X_i$ to $X_j$ in Eq. (\ref{eq:sum_prod}) is proportional to
\begin{equation}
\phi_i(x_i) \prod_{k\in {\cal N}(X_i)\setminus \{j\}} m_{k\to i}(x_i),
\end{equation}
which is in the same form of Eq. (\ref{eq:belief}) when a belief is computed.
Therefore, both messages and beliefs can be written in the form of $P(x)=\prod_{\ell=1}^{L}p_\ell(x)$,
where $L$ is the number of factors in the product.
Let $Q=\prod_{m=1}^{M}q_m(x)$ be a distribution that explains $P$ with $M$ factors and with $X$ being an end-point to be explained.
If $Q=P$, then the distance is 0 but many edges (factors) are included in $\tilde{G}$.
Starting from a uniform distribution, the search each time finds an approximating distribution $Q(x)$ by including an additional factor in $P$ to minimize $d(P,Q)$ over all distributions $P$ representing the so-far included messages and beliefs computed by BP on $G$, and over all factors (messages or priors)
that have not been included but are contributing to $P$.
Therefore, the algorithm does not directly attempt to optimize the objective $d(b_X, \tilde{b}_X)$ for the target $X$,
but does so locally: it keeps adding more factors to best explain one of the next end-points, which can already explain the target variable to some extent.

\textbf{Variants of \textbf{GE-L}}
To further speed up \textbf{GE-L} (see Fig.~\ref{fig:speedup}),
especially on graphs that a user has prior knowledge about the topology of the explaining subgraph, its search space can be further constrained.
On the one hand, we can only extend a candidate on the end-point that is added most recently, creating a chain of variables so that one is explaining the other.
This aligns with the conclusion that causal explanations are easier to be understood by the end-users~\cite{lombrozo2006structure}, and our explanations of the inference are indeed causal: how $\tilde{b}_X$ is computed by a smaller subgraph will be presented to the end-users.
On the other hand, when a target variable has many incoming messages (which is the case on social and scientific networks),
it is best to spend the explaining capacity on direct neighbors.
In the experiments,
we adopt these two constraints over \textbf{GE-L}
on the review and remaining networks, respectively.

%% file: experiment.tex
\section{Experiments}
In this section, we examine the explanation faithfulness and interpretability of \textbf{GE-L}, \textbf{GE-G}, and the state-of-the-art baselines, including \textbf{LIME}, in ten different networks from four domains (spam detection, citation analysis, social networks, bioinformatics). We also evaluate the scalability and sensitivity analysis of these methods. Moreover, we conduct user study to demostrate the usability of \textbf{GE-G}.

\subsection{Datasets}
We drew datasets from four applications.
First, we adopt the same three Yelp review networks (YelpChi, YelpNYC, and YelpZip) from~\cite{rayana2015collective} for spam detection tasks. We represent reviewers, reviews, and products and their relationships (reviewer-review and review-product connections) by an MRF. 
BP can infer the labels (suspicious or normal) of reviews on the MRF given no labeled data but just prior suspiciousness probability computed from meta-data~\cite{rayana2015collective}.
Second, in collective classification,
we construct an MRF for each one of three citation networks (Cora, Citeseer, and PubMed)
that contain papers as nodes and undirected edges as paper citation relationships~\cite{namata2011collective}.
As a transductive learning algorithm,
BP can infer the distributions of paper topics of an unlabeled paper,
given labeled nodes in the same network.
Third, we represent blogs (BlogCatalog), videos (Youtube), and users (Flickr) as nodes and behaviors including subscription and tagging as edges \citep{tang2009relational}. BP infers the preferences of users. 
The goal is to classify social network nodes into multiple classes.
Lastly, in biological networks, 
we adopt the networks analyzed in~\cite{zitnik2017predicting} which denotes nodes as protein-protein pairs and the subordination relations of protein pair as the class.
Explaining BP inference is important in all these applications: the MRFs are in general large and cyclic for a user to thoroughly inspect why a paper belongs to a certain area,
or why a review is suspicious, or why a blog is under a specific topic, 
or why two proteins connect to each other.
It is much easier for the user to verify the inference outcome on much smaller explaining graphs.
The statistics of the datasets are shown in Table~\ref{tab:datasets}.
\begin{table}[t]
    \caption{Ten networks from four application domains.}
    \label{tab:datasets}
    \scriptsize
    \centering
    \begin{tabular}{c|c|c|c|c}
    \toprule
    \textbf{Datasets} & \textbf{Classes} & \textbf{Nodes} & \textbf{Edges} & \textbf{edge/node} \\
    \midrule
    \textbf{YelpChi} & 2 & 105,659 & 269,580 & 2.55 \\
    \textbf{YelpNYC} & 2 & 520,200 & 1436,208 & 2.76 \\
    \textbf{YelpZip} & 2 & 873,919 & 2434,392 & 2.79 \\
    \midrule
    \midrule
    \textbf{Cora} & 7 & 2,708 & 10,556 & 3.90 \\
    \textbf{Citeseer} & 6 & 3,321 & 9,196 & 2.78 \\
    \textbf{PubMed} & 3 & 1,9717 & 44,324 & 2.25 \\
    \midrule
    \midrule
    \textbf{Youtube} & 47 & 1,138,499 & 2,990,443 & 2.63 \\
    \textbf{BlogCatalog} & 39 & 10,312 & 333,983 & 32.39 \\
    \textbf{Flickr} & 195 & 80,513 & 5,899,882 & 73.28 \\
    \midrule
    \midrule
    \textbf{Bioinformatics} & 144 & 13,682 & 287,916 & 21.04 \\
    \bottomrule
    \end{tabular}
\end{table}

\subsection{Experimental setting}

On Yelp review networks, a review has two and only two neighbors (a reviewer posts the review and the product receives the review), while a reviewer and a product can be connected to multiple reviews.
On the remaining networks, nodes are connected to other nodes without constraints of the number of neighbors and the type of nodes. 
We apply the two variants of GE-L on Yelp and other networks, respectively. Psychology study shows that human beings can process about seven items at a time \cite{miller1956magical}. To balance the faithfulness and interpretability, both \textbf{GE-L} and \textbf{GE-G} search subgraphs that are of maximum size five starting from the target node. In the demon,
we explore larger explaining subgraphs to allow users to select the smallest subgraph that makes the most sense.

On all ten networks,
we assume homophily relationships between any pair of nodes.
For example,
a paper is more likely to be on the same topic of the neighboring paper,
and two nodes are more likely to be suspicious or normal at the same time.
On Yelp, we set node priors and compatibility matrices according to~\cite{rayana2015collective}.
On other networks,
we assign 0.9 to the diagonal and $\frac{0.1}{c-1}$ to the rest of the compatibility metrics. 
As for priors, 
we assign 0.9 to the true class of a labeled node,
and $\frac{0.1}{c-1}$ to the remaining classes, where $c$ is the number of classes in data.
For unlabeled nodes, we set uniform distribution over classes.
On Yelp, there is no labeled node, 
and on Youtube network, 5\% are labeled nodes. 
For the remaining networks, we set the ratio of labeled data as 50\%. 
With consideration the size of the large networks, we sample 1\% from the unlabeled nodes as target nodes on Youtube and Flickr datasets, and sample 20\% of unlabeled nodes as target nodes on BlogCatalog and Bioinformatics networks.

\begin{table*}[!htbp]
\caption{Overall performance: in general \textbf{Comb} is the best method that significantly outperforms all other methods on all networks.
We use $\bullet$ to indicate whether statistically \textbf{GE-L} significantly (pairwise t-test at 5\% significance level) outperforms \textbf{Random} and whether \textbf{Comb}
outperforms \textbf{GE-G} ($k$=3), respectively.
}
\label{tab:overall}
\centering
\footnotesize
\begin{tabular}{ c | c c | c c | c c c c}
\toprule

\textbf{Method} & 
\textbf{Embedding} & \textbf{LIME} & 
\textbf{Random} & 
\textbf{GE-L} & 
\textbf{Random} & 
\textbf{GE-G ($k$=1)} & 
\textbf{GE-G ($k$=3)} & \textbf{Comb} \\

\midrule
Chi
& \makecell[c]{0.058[5.0]}
& \makecell[c]{5.300} 
& \makecell[c]{0.053[3.9]$\bullet$}
& \makecell[c]{0.036[3.9]} 
& \makecell[c]{0.022[5.0]} 
& \makecell[c]{0.0012[5.0]}
& \makecell[c]{0.0012[5.0]$\bullet$}
& \makecell[c]{\textbf{0.0006}[6.5]}
\\
NYC
& \makecell[c]{0.084[5.0]}
& \makecell[c]{5.955} 
& \makecell[c]{0.043[4.1]$\bullet$}
& \makecell[c]{0.028[4.1]} 
& \makecell[c]{0.017[5.0]} 
& \makecell[c]{0.0012[5.0]}
& \makecell[c]{0.0011[5.0]$\bullet$}
& \makecell[c]{\textbf{0.0006}[6.2]}
\\
Zip
& \makecell[c]{0.084[5.0]}
& \makecell[c]{6.036} 
& \makecell[c]{0.040[4.2]$\bullet$}
& \makecell[c]{0.025[4.2]} 
& \makecell[c]{0.010[5.0]} 
& \makecell[c]{0.0014[5.0]}
& \makecell[c]{0.0013[5.0]$\bullet$}
& \makecell[c]{\textbf{0.0008}[6.1]}
\\

\midrule
\midrule
Cora
& \makecell[c]{0.527[4.8]}
& \makecell[c]{1.321}
& \makecell[c]{0.362[3.6]$\bullet$}
& \makecell[c]{0.181[3.6]}
& \makecell[c]{0.594[4.8]}
& \makecell[c]{0.137[4.8]}
& \makecell[c]{0.132[4.9]$\bullet$}
& \makecell[c]{\textbf{0.084}[6.4]}
\\
Citeseer 
& \makecell[c]{0.305[4.4]}
& \makecell[c]{1.221}
& \makecell[c]{0.243[2.9]$\bullet$}
& \makecell[c]{0.108[2.9]}
& \makecell[c]{0.340[4.4]}
& \makecell[c]{0.077[4.4]}
& \makecell[c]{0.075[4.4]$\bullet$}
& \makecell[c]{\textbf{0.048}[5.7]}
\\
PubMed
& \makecell[c]{0.842[5.0]} 
& \makecell[c]{0.910}
& \makecell[c]{0.718[3.1]}
& \makecell[c]{0.577[3.1]}
& \makecell[c]{0.893[5.0]}
& \makecell[c]{0.188[5.0]}
& \makecell[c]{0.185[5.0]$\bullet$}
& \makecell[c]{\textbf{0.098}[7.1]}
\\
\midrule
\midrule
Youtube
& \makecell[c]{0.340[5.0]}
& \makecell[c]{-}
& \makecell[c]{0.376[2.8]}
& \makecell[c]{0.321[2.8]}
& \makecell[c]{0.343[5.0]}
& \makecell[c]{0.263[5.0]}
& \makecell[c]{0.264[5.0]}
& \makecell[c]{\textbf{0.225}[6.7]}
\\
Flickr  
& \makecell[c]{5.903[5.0]}
& \makecell[c]{-}
& \makecell[c]{6.259[4.7]}
& \makecell[c]{6.232[4.7]}
& \makecell[c]{6.018[5.0]}
& \makecell[c]{4.654[5.0]}
& \makecell[c]{4.652[5.0]}
& \makecell[c]{\textbf{4.111}[7.4]}
\\
BlogCatalog
& \makecell[c]{7.887[5.0]}
& \makecell[c]{-}
& \makecell[c]{7.899[4.8]}
& \makecell[c]{8.054[4.8]}
& \makecell[c]{7.867[5.0]}
& \makecell[c]{6.621[5.0]}
& \makecell[c]{6.702[5.0]}
& \makecell[c]{\textbf{6.343}[7.9]}
\\
\midrule
\midrule
Bioinformatics 
& \makecell[c]{2.065[5.0]}
& \makecell[c]{-}
& \makecell[c]{2.085[4.9]}
& \makecell[c]{1.893[4.9]}
& \makecell[c]{2.116[5.0]}
& \makecell[c]{1.423[5.0]}
& \makecell[c]{1.508[5.0]}
& \makecell[c]{\textbf{1.356}[5.6]}
\\
\bottomrule
\end{tabular}
\end{table*}

\subsection{Baselines}

\textbf{Random} It ignores the messages computed by BP and selects a node in $\partial\tilde{G}^{(t-1)}$ randomly when extending $\tilde{G}^{(t-1)}$. To compare fairly, \textbf{Random} searches the subgraph of the same structure as those found by \textbf{GE-L} and \textbf{GE-G}, respectively.

\textbf{Embedding} It constructs subgraphs with the same size as those found by \textbf{GE-G}. 
However, it utilizes DeepWalk\cite{perozzi2014deepwalk} to obtain node embeddings, 
based on which top candidate nodes similar to the target are included to explain the target variable.

\textbf{LIME} \cite{ribeiro2016should} It is the state-of-the-art
black-box explanation method that works for classification models when input data are vectors rather than graphs.
We randomly select 200 neighbors of each target node in the node feature vector space, with sampling probability weighted by the cosine similarity between the neighbors and the target. The feature vector is defined either as \cite{rayana2015collective}
(on Yelp review networks) or bag-of-words (on the citation networks).
A binary/multiclass logistic regression model is then fitted on the sample and used to approximate the decision boundary around the target.
\textbf{LIME} is less efficient than the subgraph search-based approaches since a new classification model has to be fitted for each target variable.
\textbf{LIME} explains 30\% of all review nodes, randomly sampled on Yelp review datasets, 
and explain all unlabeled variables on the citation networks.
It cannot explain nodes in the remaining four networks due to the lack of feature vectors, which is one of the drawbacks of \textbf{LIME}.

\textbf{Comb} It aggregates all candidate subgraphs from \texttt{Beam[C]} into a single graph as an explanation.
This method can aggregate at most $k$ (beam size) subgraphs and have at most $kC$ variables.
Here, we set $k$=3 and report the performance of the combined subgraph.
The performances of the top candidate in \texttt{Beam[C]} with $k$=1 and $k$=3 are reported under \textbf{GE-G} ($k$=1)
and \textbf{GE-G} ($k$=3).

\subsection{Explanation Accuracy}


\noindent\textbf{Overall Performance}
Quantitative evaluation metrics
For each method, we run BP on the extracted subgraph $\tilde{G}$ for each target variable $X$ to obtain a new belief $\tilde{b}_X$ (except for \textbf{LIME} that does not construct subgraphs).
Explanation faithfulness is measured by the objective function in Eq. (\ref{eq:opt}).
In Table~\ref{tab:overall}, we report the mean of the performance metric overall target variables, and the best methods are boldfaced. We also report the average size of explaining subgraphs in the square brackets after individual means.

From Table~\ref{tab:overall}, we can conclude that: 
1) \textbf{Comb} always constructs the largest subgraph and performs best, due to multiple alternative high-quality explanations from the beam branches.
2) \textbf{GE-G} ($k$=3) is the runner up and better than \textbf{GE-G} ($k$=1), because the search space of \textbf{GE-G} ($k$=3) is larger than \textbf{GE-G}($k$=1)'s. 
3) The performance of \textbf{Embedding} is not very good, but still better than \textbf{LIME} in all cases. \textbf{LIME} has the worst performance, as it is not designed to explain BP and cannot take the network connections into account.
4) Faithfulness is positively correlated to subgraph size. 

\noindent\textbf{Spam detection explanations}
On Yelp review networks,
\textbf{GE-L} generates chain-like subgraphs. 
The average subgraph size is around four,
even though the maximum capacity is five. 
This is because \textbf{GE-L} focuses on local information only and stops early when the prior of the last added node best explains the previously added message.
Both \textbf{GE-G} versions extend the subgraph to the maximum size to produce a better explanation.
Notice that \textbf{Random} performs better when imitating \textbf{GE-G} ($k$=1) than when imitating \textbf{GE-L}. 
The reason is that there are only two types of neighboring nodes of the target node
and \textbf{Random} 
imitating \textbf{GE-G} has a higher chance to include the better neighbor and also generates larger subgraphs.



\noindent\textbf{Collective classification tasks}
In these tasks, \textbf{GE-L} constructs star-like subgraphs centered at the target node.
On Cora and Citeseer, the performance of \textbf{GE-L} is closer to (but still inferior to) \textbf{GE-G} with both beam sizes,
compared to their performance difference on Yelp.
This is because the Cora and Citeseer networks consist of many small connected components, most of which contain less than five nodes, essentially capping \textbf{GE-G}'s ability to add more explaining nodes to further bring up the faithfulness (evidenced by the average size of subgraphs found by \textbf{GE-G} smaller than 5). 
Compared with Yelp review networks, interestingly, \textbf{Random} imitating \textbf{GE-G} generates larger subgraphs but performs worse than \textbf{Random} imitating \textbf{GE-L}. 
The reason is that \textbf{GE-G} can add nodes far away from the target node and the random imitation will do the same.
However, without the principled guidance in \textbf{GE-G},
\textbf{Random} can add nodes those are likely in the other classes than the class of the target node.


\subsection{Scalability}
\label{sec:scalability}

\begin{figure}[]
\centering
\begin{minipage}{.25\textwidth}
    \label{fig:parallel_Chi}
    \includegraphics[width=\textwidth]{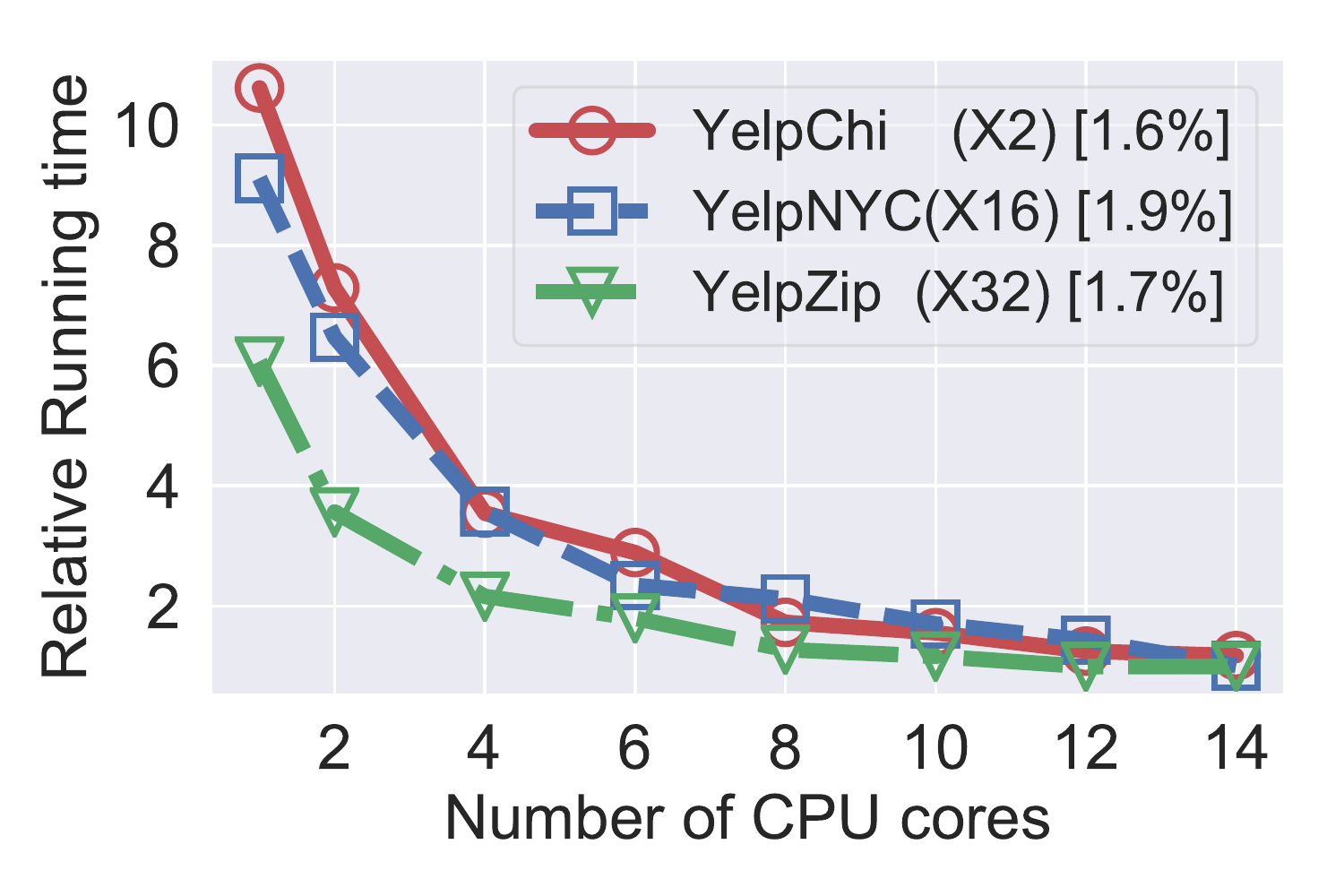}
\end{minipage}%
\begin{minipage}{.25\textwidth}
    \label{fig:parallel_Chi}
    \includegraphics[width=\textwidth]{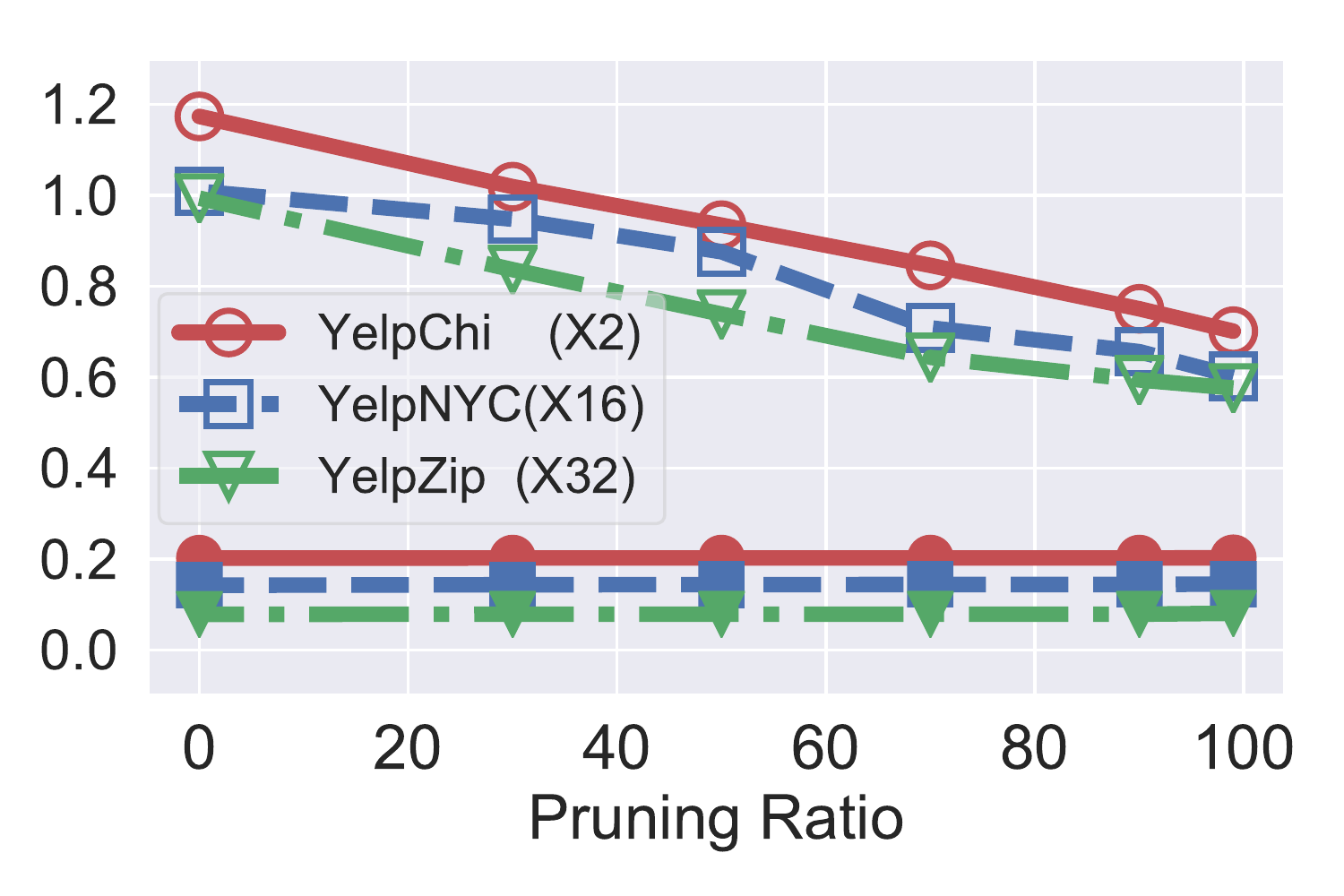}
    
\end{minipage}%
\caption{
    \textbf{Left} the computing time reduces superlinearly as the number of cores increases from 1 to 14. 
    The parentheses enclose the number of hour(s) per unit.
    The square brackets enclose the running time of \textbf{GE-L}. 
    For example, \textbf{GE-G} takes about 192 hours using one core on YelpZip, while \textbf{GE-L} costs 3 hours.
    \textbf{Right} Effect of pruning: as we increase the pruning rate to 99\%,
    the subgraph explanation faithfulness is not degraded (shown by the 3 horizontal lines with solid markers),
    while the running time is reduced by 1/3.}
\label{fig:speedup}
\end{figure}

We run \textbf{GE-G} ($k$=1) on Yelp review networks which have the most nodes to be explained.
The scalability of the algorithm is demonstrated in two aspects in Fig.~\ref{fig:speedup}.
First, 
the search can be done in parallel on multiple cores.
As we increase the number of cores, the running time goes down super-linearly as the number of cores increases from 1 to 14.
Second,
candidate pruning plays a role in speeding up the search.
We use 14 cores and increase the pruning ratio from 0 to 99\% and obtain about $\times 1.7$ speedup at most.
Importantly, the explanation faithfulness is not affected by the pruning, shown by the three lines at the bottom of the right figure (200 $\times$ mean objective).

\begin{figure*}[t]
\centering
\begin{minipage}{.24\textwidth}
\label{fig:cora_kl}
    \includegraphics[width=1\textwidth]{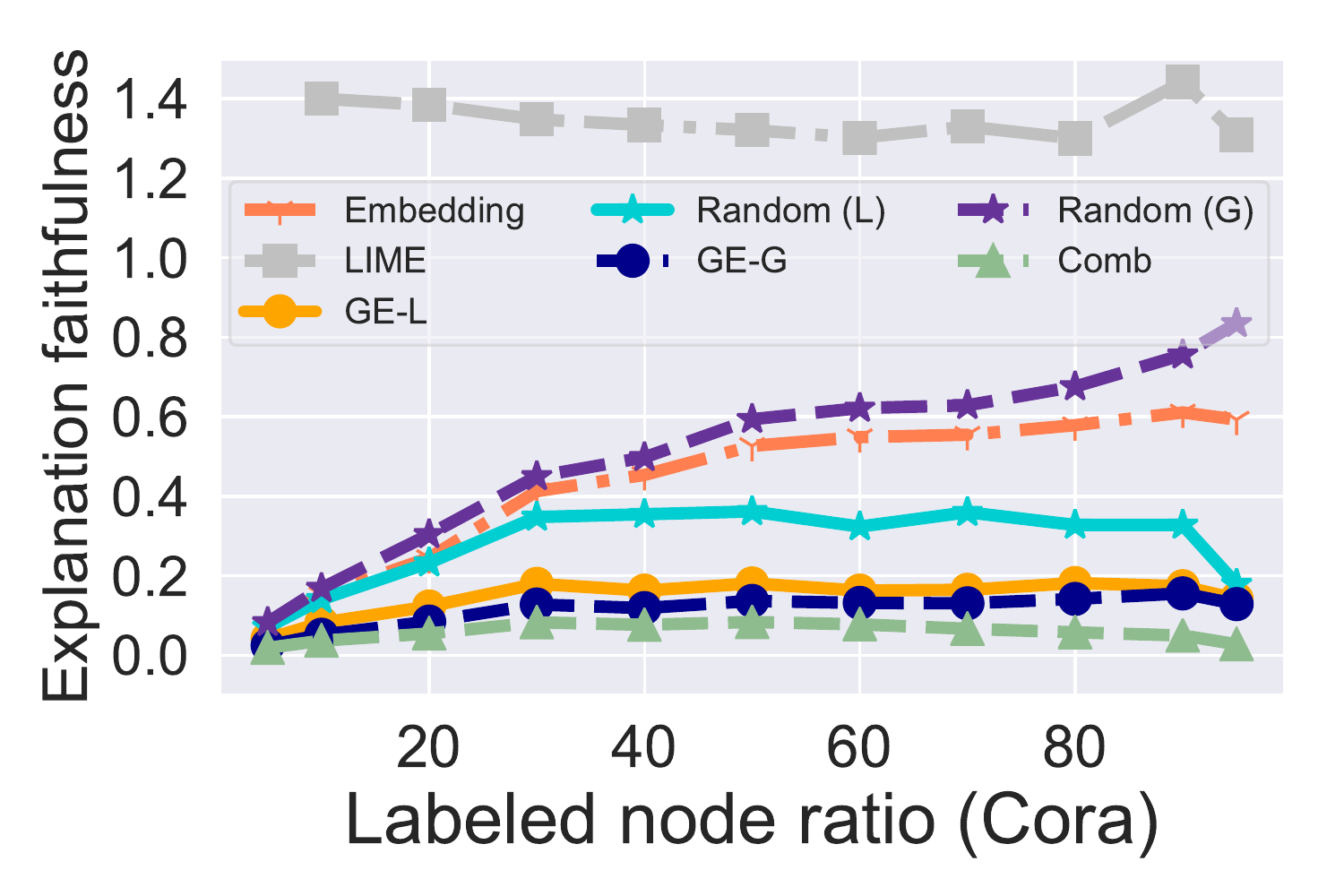}
\end{minipage}%
\begin{minipage}{.24\textwidth}
\label{fig:citeseer_kl}
    \includegraphics[width=1\textwidth]{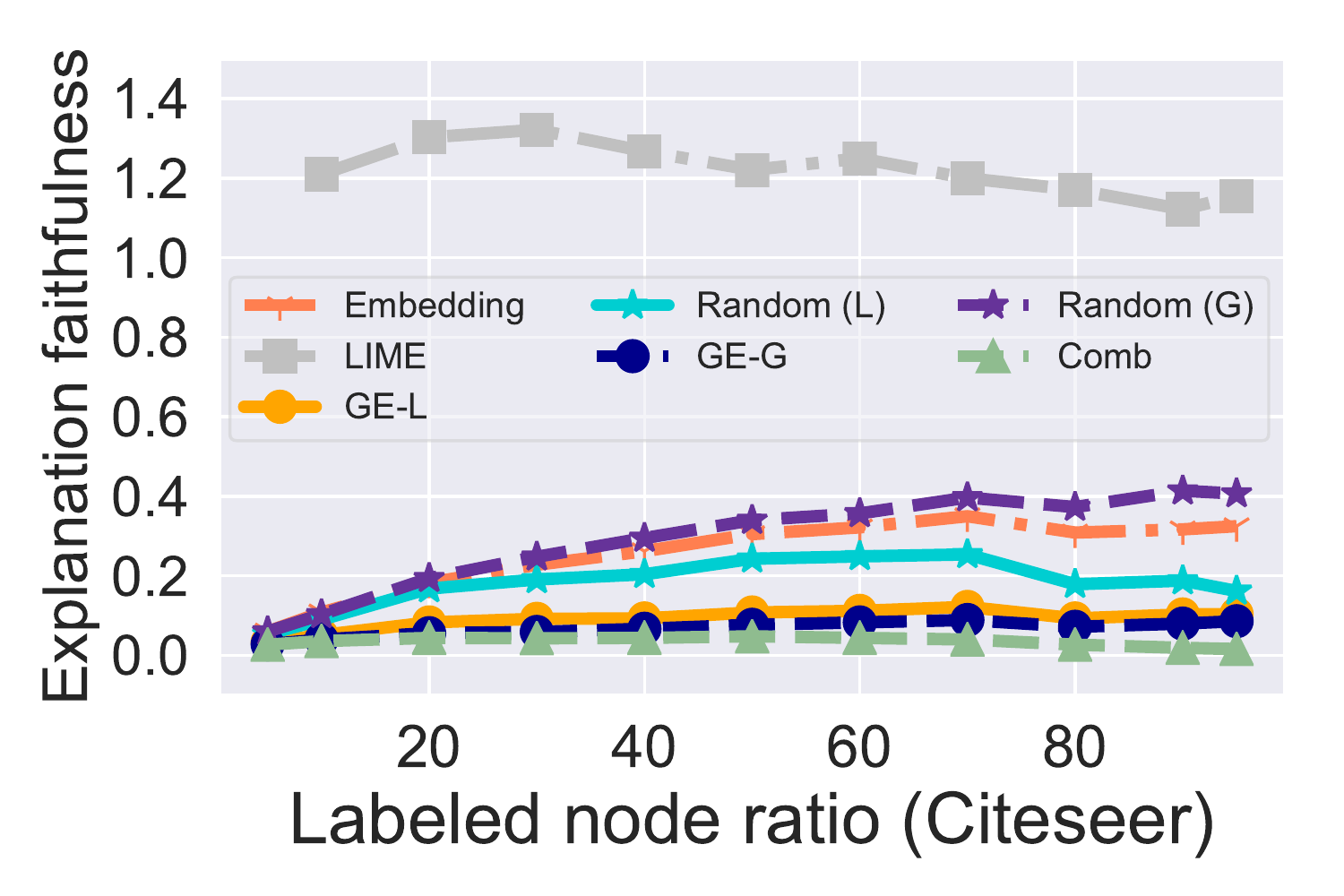}
\end{minipage}%
\begin{minipage}{.24\textwidth}
\label{fig:pubmed}
    \includegraphics[width=1\textwidth]{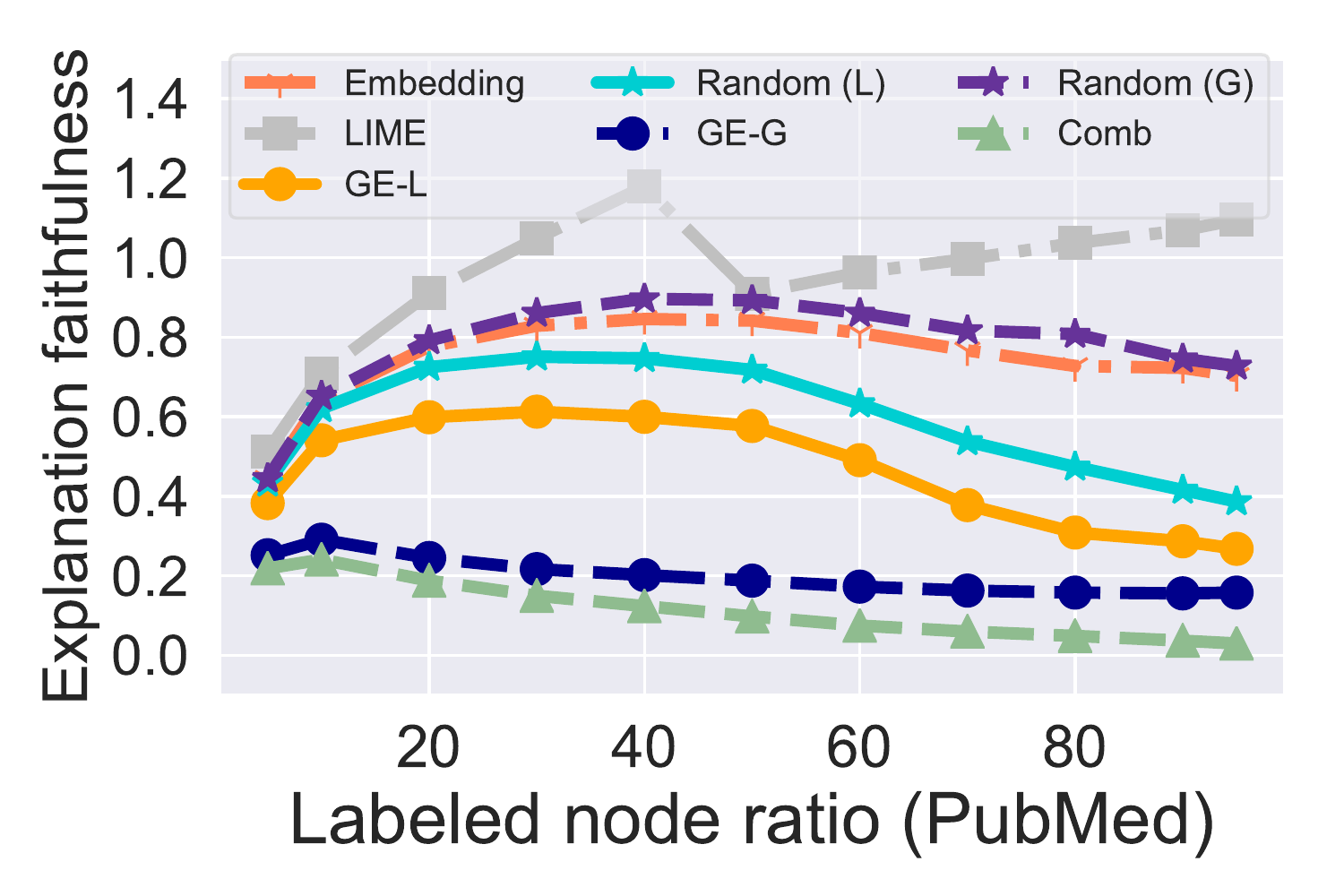}
\end{minipage}%
\begin{minipage}{.24\textwidth}
\label{fig:protein_kl}
    \includegraphics[width=1\textwidth]{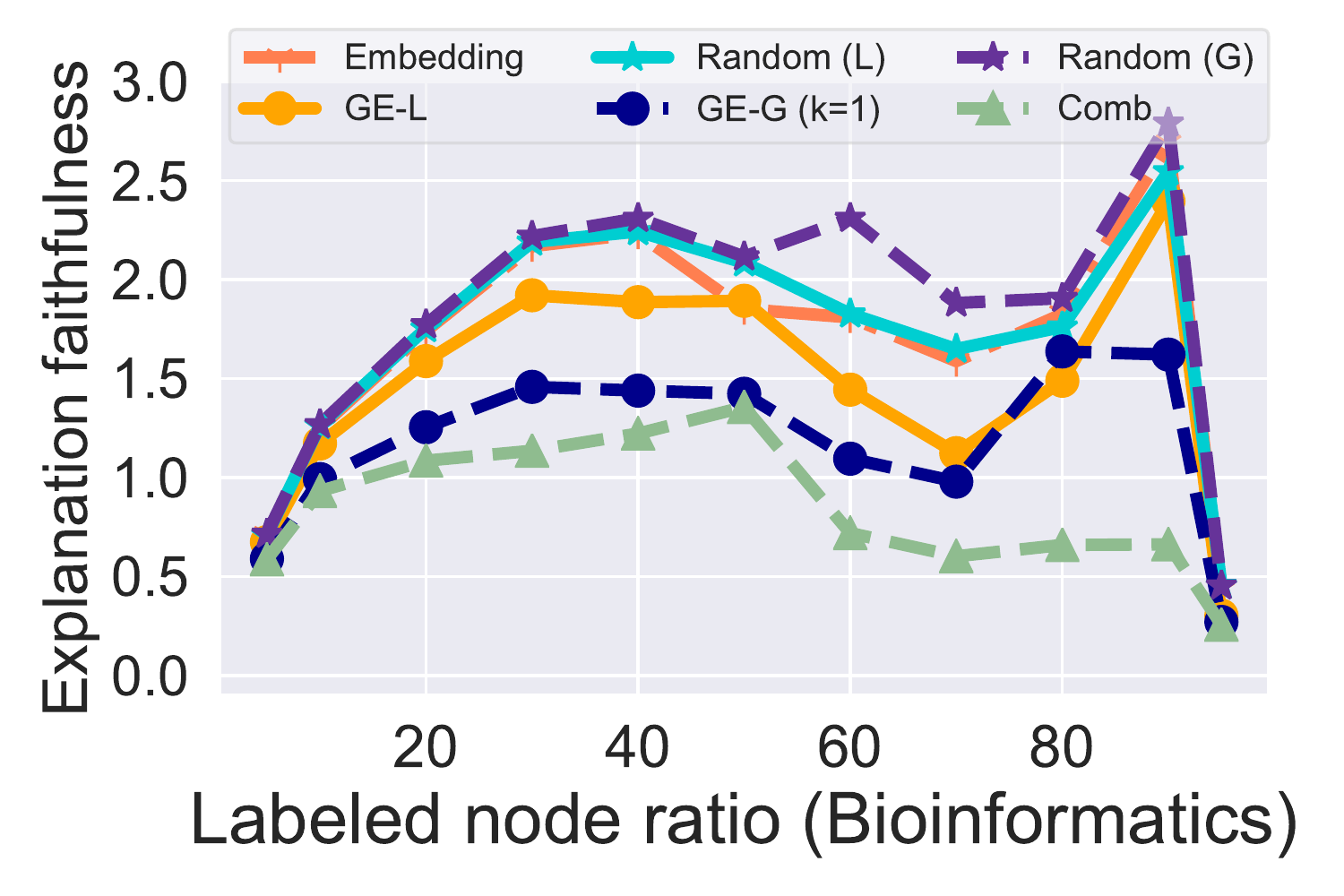}
\end{minipage}%
    \caption{\small 
        Explanation faithfulness (the smaller the better) on citation networks while the ratio of labeled nodes increases. \textbf{Random (L)} is the random baseline that imitate \textbf{GE-L}. Similar for \textbf{Random (G)}.
    }
    \label{fig:citation_overall}
\end{figure*}

\begin{figure*}[t]
\centering
\begin{minipage}{.24\textwidth}
\label{fig:size_zip}
    \includegraphics[width=1\textwidth]{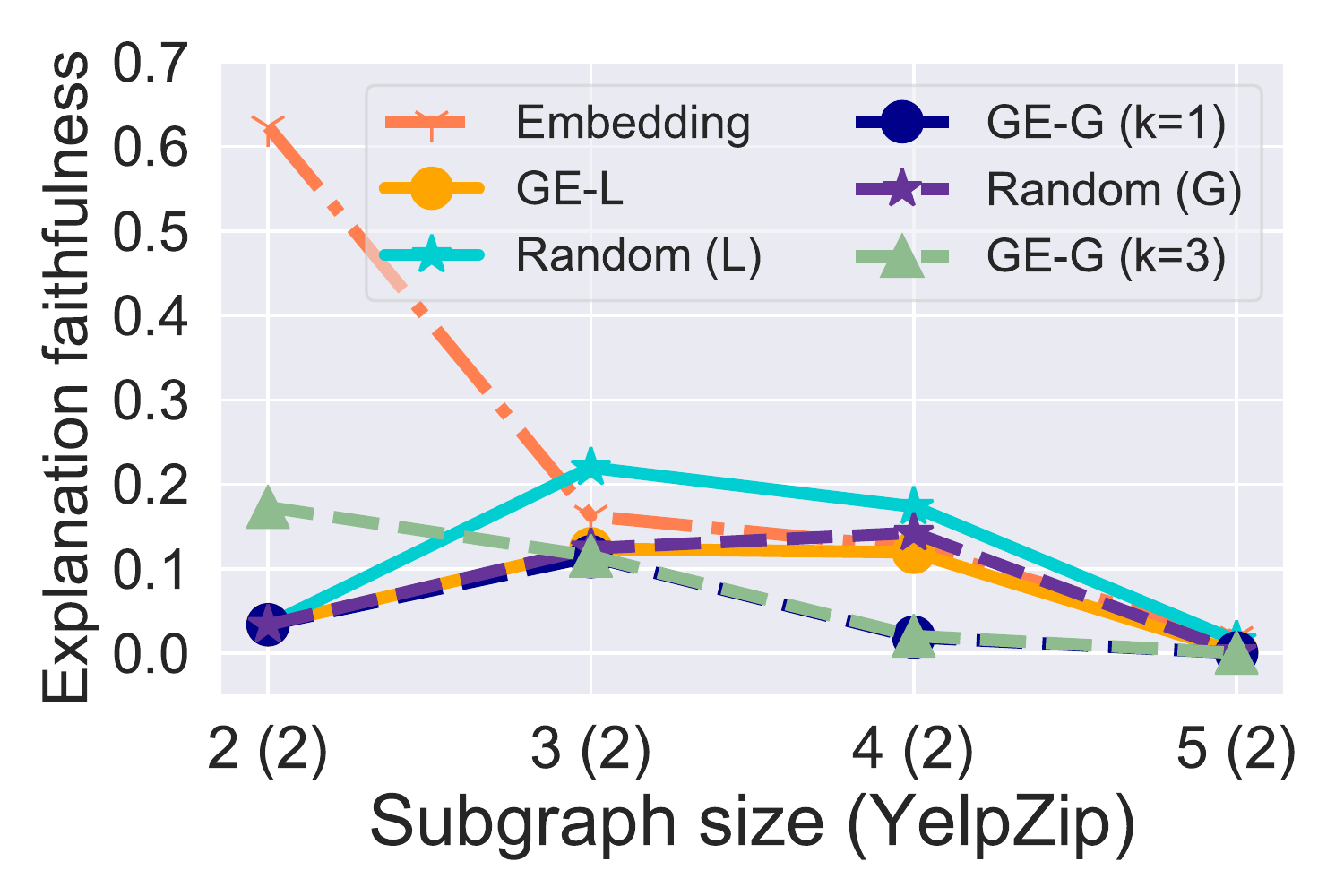}
\end{minipage}%
\begin{minipage}{.24\textwidth}
\label{fig:size_pubmed}
    \includegraphics[width=1\textwidth]{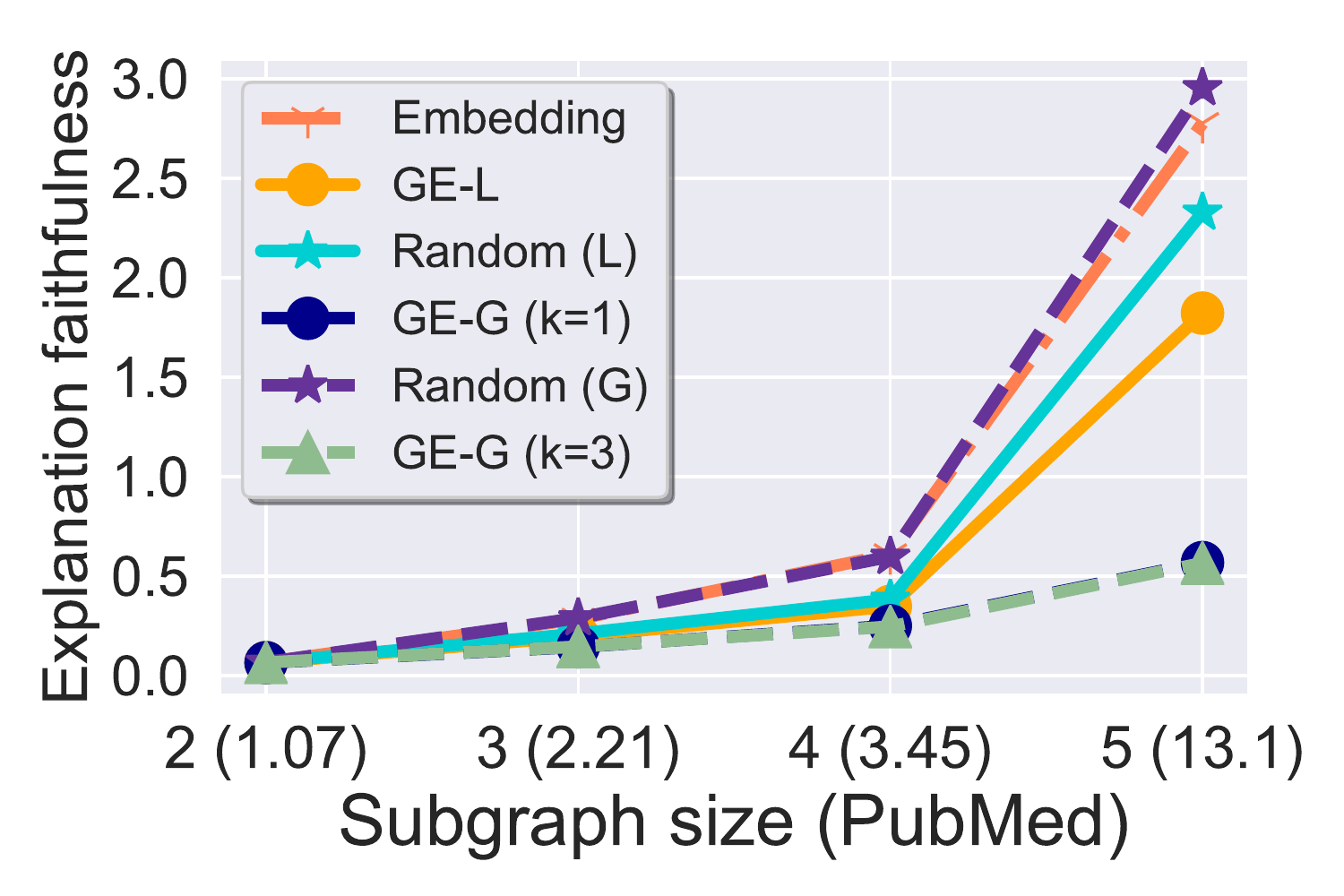}
\end{minipage}
\begin{minipage}{.24\textwidth}
\label{fig:size_blog}
    \includegraphics[width=1\textwidth]{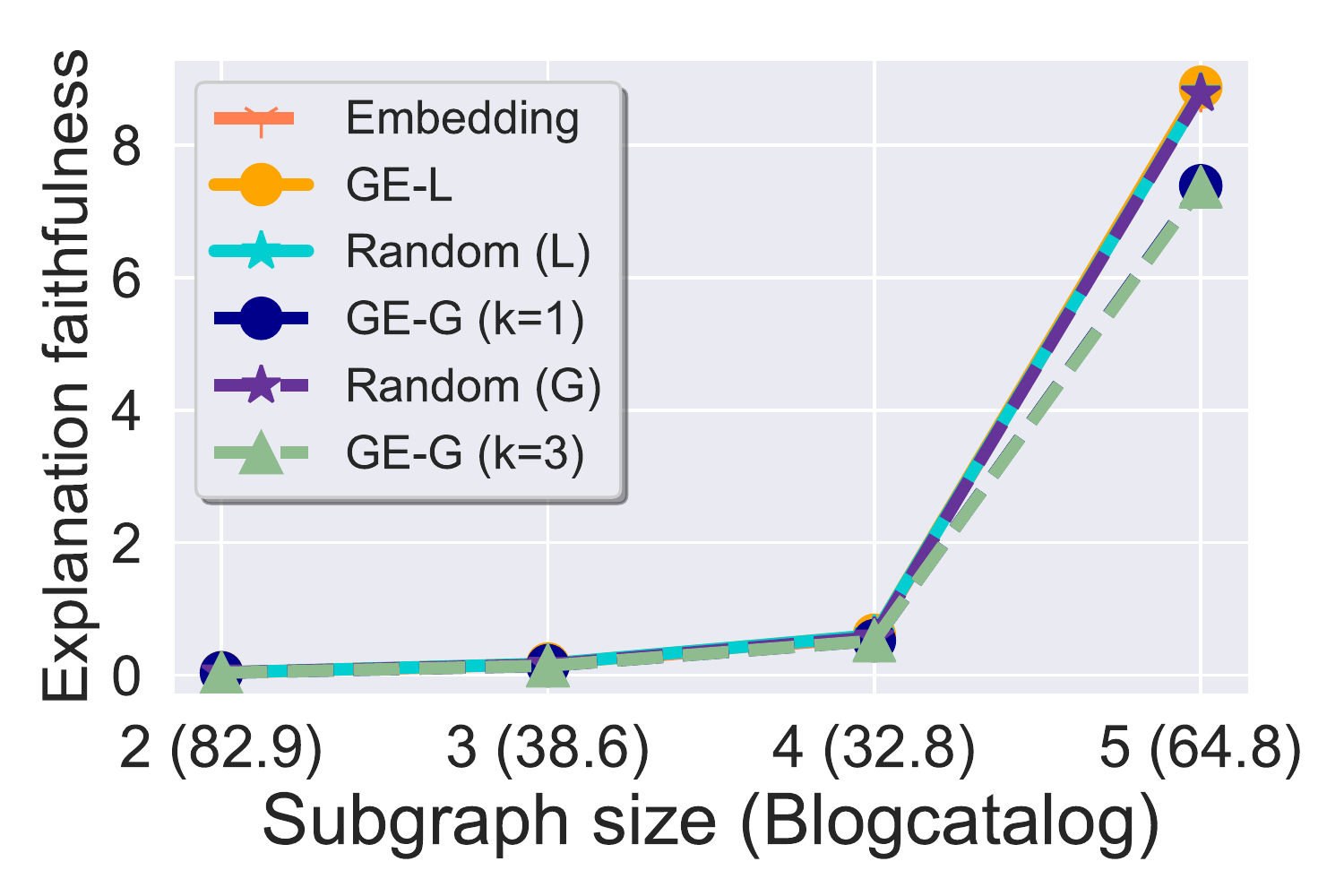}
\end{minipage}
\begin{minipage}{.24\textwidth}
\label{fig:size_bio}
    \includegraphics[width=1\textwidth]{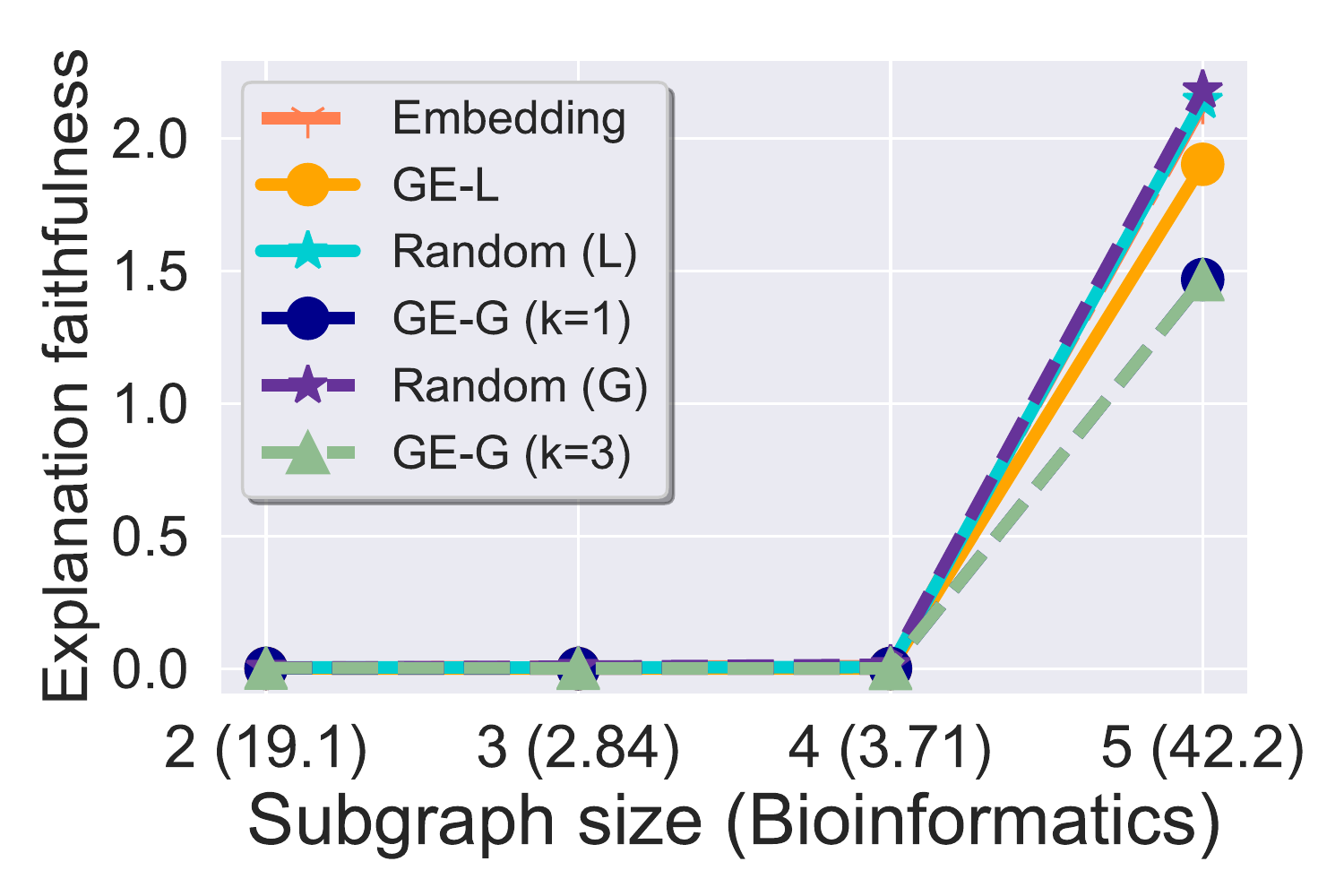}
\end{minipage}
    \caption{\small 
        Relationship between size of subgraph and faithfulness
    }
    \label{fig:size_sen}
\end{figure*}

\subsection{Sensitivity analysis}
There are two hyperparameters for the subgraph search and we study the sensitivity of the explanation faithfulness with respect to these two parameters.
First, when there are labeled nodes on $G$,
the explanation subgraphs may include a labeled node in the explanation,
and the more labeled nodes, the more likely the explanation subgraphs will include these nodes.
The question here is whether including such labeled nodes will improve explanation faithfulness, and does our method require a large number of labeled nodes.
To answer the questions,
on the four networks (Cora, Citeseer, PubMed, and Bioinformatics),
we vary the ratio of labeled nodes 
in the ranges of \{5\%, 10\%, 20\%, 30\%, 40\%, 50\%, 60\%, 70\%, 80\%, 90\%, 95\%\}.
The explanation performances are shown in Fig.~\ref{fig:citation_overall}.
Overall different ratios on all four networks, \textbf{Comb} is the best and \textbf{GE-G} is runner-up. 
On the Cora and Citeseer networks, the performances of all methods are not very sensitive to the ratio of labeled nodes.
On the PubMed network, there are improvements in most baselines except \textbf{LIME} becomes worse.
One can also conclude that \textbf{LIME}, designed for parametric models including deep models, is outperformed by methods designed specifically for MRF.
On the Bioinformatics, \textbf{LIME} is not available and \textbf{GE-G} and \textbf{Comb} have better performance.

Second, subgraphs containing more nodes perform better since they have more contexts of the target.
To evaluate the sensitivity of faithfulness with respect to the subgraph size,
we downsize the subgraphs found by other methods (\textbf{embedding}, \textbf{GE-G} ($k$=1), \textbf{GE-G} ($k$=3), \textbf{Random} imitating \textbf{GE-L} and \textbf{GE-G}) to the same size of those found by \textbf{GE-L},
which may stop growing subgraphs before reaching the cap.
The results obtained on four networks are shown in Fig.~\ref{fig:size_sen}. 
On Yelp,
\textbf{GE-L}, \textbf{GE-G}, and \textbf{Random} perform the same when the size is two
since, with this size, the imitating \textbf{Random} 
has to adopt the same subgraph topology and node type as the subgraph found by \textbf{GE-L} and \textbf{GE-G}.
As we downsize the best subgraphs found by \textbf{GE-G} ($k=3$),
a better subgraph with a large size
may not be optimal when trimmed to size two,
lacking the optimal substructures that facilitate dynamic programming variable can alter the target belief, leading to more insight into and trust on the BP.

\subsection{Explanation visualization and utility}
\label{sec:user_study}
\textbf{Explanation for Ratification}
One of the goals of explanations is to help end-users
understand why the probabilistic beliefs are appropriate given the MRF,
so that they can develop confidence in the beliefs and BP.
This process is called ``ratification''~\cite{suermondt1993explanation}.
To attain this goal, Fig.~\ref{fig:userstudy} displays 
multiple explanations (trees) generated by \textbf{Comb} on the Karate Club network,
with increasing size in Fig.~\ref{fig:userstudy} along with the faithfulness measured by the distance from Eq. (\ref{eq:opt}).
One can see that
the distance decreases exponentially fast as more nodes are added to the explanation.
On each explaining subgraph (tree),
we display beliefs found by BP on $G$ and $\tilde{G}$ so that the user is aware of the gap. 
Since insight is personal~\cite{suermondt1993explanation},
the interface is endowed with flexibility for a user to navigate through alternatives, 
and use multiple metrics
(distributional distance, subgraph size and topology) to select the more sensible ones.
This design also allows the user to see how adding or removing an variable can alter the target belief, leading to more insight into and trust on the BP.

\textbf{Explanation for Verification}
MRF designers can use the explanations for verification of their design,
including network parameters and topology.
We demonstrate how to use the generated subgraphs to identify a security issue in spam detection.
Specifically,
we desire that a spam review should be detected with equal probability by the same spam detector, regardless of how many reviews their authors have posted.
On the one hand, it has been found that 
well-camouflaged elite spammer accounts can deliver better rating promotion~\cite{luca2016reviews}, a fact that has been exploited by dedicated spammers~\cite{zheng2017smoke}.
On the other hand,
a spam detector can be accurate in detecting fake reviews for the wrong reason,
such as the prolificacy of the reviewers~\cite{ross2017right}.

We gather all reviews from the YelpChi dataset,
and from the generated subgraph explanations, create four features, including whether a review is connected to its reviewer and target product, and the degree of its two potential neighbors.
We then train a logistic regression model on these features to predict the probability of a review being a false positive (4,501), a false negative (6,681), and mis-classified (11,182), using predictions based on SpEagle~\cite{rayana2015collective}.
The point is that the above mistakes should not depend on the prolificacy of the connected reviewers and products.
However, we found that a sizable subset of false negatives (1,153 out of 6,681) are due to the influence from the review node only based on our explanations.
Moreover, the logistic model shows that the influence of the degree of the neighboring reivewer contributes the most to the probability of FN.
This is a serious security issue: a spam from an prolific reviewers is more likely to be missed by SpEagle.





\begin{figure}[t]
  \centering
  \includegraphics[width=0.45\textwidth]{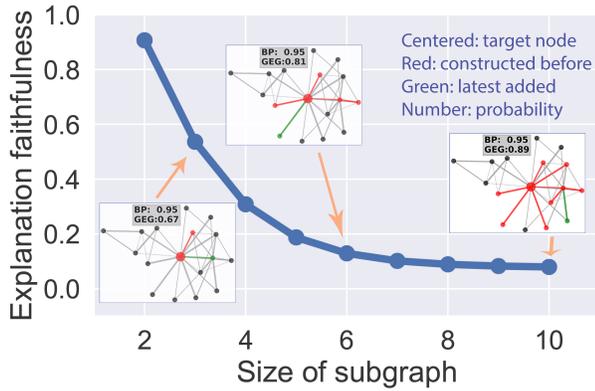}
  \caption{\small
Larger
\textcolor{red}{red} nodes are the nodes explained.
Other \textcolor{red}{red} nodes are the explaining nodes.
The \textcolor{green}{green} nodes are the newly added ones.
  }
  \label{fig:userstudy}
\end{figure}

%% file: related_work.tex
\section{Related work}
\label{sec:related}

To explain differentiable parametric predictive models, such as deep networks~\cite{koh2017understanding} and linear models~\cite{lou2012intelligible,wojnowicz2016influence},
the gradients of the output with respect to the parameters and input data~\cite{simonyan2013deep}
can signify key factors that explain the output.
However, graphical models aim to model long range and more complicated types of interaction among variables.
If a model is too large or complex to be explained,
an approximating model can provide certain amount of transparency to the full model.
In~\cite{ribeiro2016should, baehrens2010explain},
parametric or non-parametric models are fitted to approximate a more complex model locally.
The idea of approximation is similar to that in GraphExp, with different approximation loss functions.
We have seen in the experiments that a parametric model won't deliver good explanations to the inference outcomes on a graphical model.
In~\cite{krakovna2016increasing}, HMM is trained to approximate an RNN model, resembling the idea of using a simpler grapical model to approximate a more complex graphical model.
However, both HMM and RNN are linear while GraphExp focuses on graphs with more general topology, including cycles.

Explainable bayesian networks were studied in the 1980's and 1990's~\cite{norton1988explanation, suermondt1993explanation},
driven by the needs to verify and communicate the inference outcomes of Bayesian networks in expert systems.
It has long been recognized that
human users are less likely to adopt expert systems without interpretable and probable explanations~\cite{teach1981analysis}.
More recently,
Bayesian networks were formulated as a multi-linear function so that explanations can be facilitated by differentiation~\cite{darwiche2003differential}.
The fundamental difference between GraphExp and these prior works is that we handle MRFs with cycles while they handled Bayesian networks without cycles.
The differentiation-based explanation of MRFs in~\cite{chan2005sensitivity}
finds a set of important netowrk parameters (potentials)
to explain changes in the marginal distribution of a target variable
without explaining any inference procedure.
GraphExp generates graphical models consisting of prominent variables
for reproducing the inference of a BP procedure on a larger graph.
Interpretable graphical models
are also studied under the hood of topic models~\cite{chang2009reading},
where
the focus is to communicate the meaning of the inference outcomes through
measures or prototypes (which words belong to a topic),
rather than explaining how the outcomes are arrived at.
